\documentclass[conference, 10pt]{IEEEtran}
\IEEEoverridecommandlockouts

\newif\iflong
\longtrue

\usepackage[normalem]{ulem}

\usepackage{microtype}
\usepackage{amsmath,amssymb,amsfonts,amsthm,onlyamsmath}
\usepackage{graphicx}
\usepackage[usenames]{xcolor}
\usepackage[capitalize,noabbrev]{cleveref}
\usepackage{balance}
\usepackage{nicefrac}
\usepackage{mathtools}
\usepackage{tikz}
\usepackage{thmtools}
\usepackage{bm}
\usepackage{pgf}

\usepackage[style=ieee, citestyle=numeric-comp,sorting=nyt, sortcites,backend=biber,maxbibnames=99,
maxcitenames=2,
mincitenames=1,
maxalphanames=4,
minalphanames=3,
doi=false
]{biblatex}
\AtNextBibliography{\footnotesize}
\DefineBibliographyStrings{english}{%
    andothers = {et\addabbrvspace al\adddot}
}
\DeclareFieldFormat{sentencecase}{#1}

\def\nobreakbefore{\relax
  \ifvmode\else
    \ifhmode
      \ifdim\lastskip > 0pt\relax
        \unskip\nobreakspace
      \fi
    \fi
  \fi
}
\let\oldcite\cite
\renewcommand\cite{\nobreakbefore\oldcite}

\defcounter{abbrvpenalty}{9000}

\DeclareSourcemap{
  \maps[datatype=bibtex]{
    \map{
      \step[fieldsource=url,
            match=doi,
            final]
      \step[fieldset=url, null]
    }
  }
}

\AtEveryBibitem{
 \ifentrytype{inproceedings}{%
  \clearfield{series}
  \clearfield{number}
  \clearfield{volume}
 }{}%
 \clearlist{address}%
 \clearfield{issn}%
 \clearlist{location}%
 \clearfield{month}%
 \ifentrytype{book}{}{
  \clearlist{publisher}%
  \clearname{editor}%
 }%
}

\usepackage{float}
\usepackage{listings}

\floatname{algoflt}{Protocol}
\floatstyle{ruled}
\newfloat{algoflt}{tbp}{loa}
\makeatletter
\lst@AddToHook{OnEmptyLine}{\addtocounter{lstnumber}{-1}\vspace{-0.5\baselineskip}}
\newcommand{\raisedrule}[2][0em]{\leaders\hbox{\rule[#1]{1pt}{#2}}\hfill}
\def\sepline#1{\addtocounter{lstnumber}{-1}\text{\hspace{-15pt}\parbox{30pt}{\large\color{white}\mbox{}\raisedrule{2ex}~}\hspace{-25pt}\parbox{\textwidth}{\mbox{}\color{lightgray}\raisedrule{1ex}~\color{gray}#1~\color{lightgray}\raisedrule{1ex}\mbox{}}}}
\def\dcmnumberstyle{}
\def\preAlgo{}

\lstnewenvironment{algorithm}[2][htb]{%
\lstset{%
    mathescape,%
    tabsize=2,%
    numbers=left,%
    numberstyle=\tiny\dcmnumberstyle,%
    numberblanklines=false,%
    showlines=true,%
    xleftmargin=10pt,
    numbersep=5pt,
    basicstyle=\rmfamily\small,%
    backgroundcolor={},%
    columns=fullflexible,%
    emph={%
        [1]if, else, then, and, or, for, do, while, return, exit, not, output, initialize, var, execute, let %
        },
    emphstyle={%
        [1]\bfseries%
        },%
    escapeinside={/*}{*/},
    escapebegin={\ifmmode\else\hfill\color{gray}\footnotesize$\triangleright$~\fi},
    escapeend={\ifmmode\else\fi},
    }%
\setlength{\topsep}{0pt}%
\setlength{\parindent}{0pt}%
\algoflt[#1]\caption{#2}\preAlgo\gdef\preAlgo{}%
}{\endalgoflt
}

\crefname{algoflt}{Protocol}{Protocols}

\makeatother

\def\Vardef#1{%
    \expandafter\newcommand\csname #1\endcsname[1]{%
        \def\first{##1}%
        \def\second{*}%
        \def\third{}%
        \ensuremath{\mathsf{\MakeLowercase #1}\ifx\first\second\else\ifx\first\third\else[{##1}]\fi\fi}%
    }%
}
\def\KWdef#1{%
    \expandafter\newcommand\csname #1\endcsname{%
    \mathsf{\MakeLowercase #1}%
    }%
}

\let\originalleft\left
    \let\originalright\right
    \renewcommand{\left}{\mathopen{}\mathclose\bgroup\originalleft}
    \renewcommand{\right}{\aftergroup\egroup\originalright}

\makeatletter
\newcommand{\DeclareMathActive}[2]{%
  \expandafter\edef\csname keep@#1@code\endcsname{\mathchar\the\mathcode`#1 }
  \begingroup\lccode`~=`#1\relax
  \lowercase{\endgroup\def~}{#2}%
  \AtBeginDocument{\mathcode`#1="8000}%
}

\newcommand{\std}[1]{\csname keep@#1@code\endcsname}
\patchcmd{\newmcodes@}{\mathcode`\-\relax}{\std@minuscode\relax}{}{\ddt}
\AtBeginDocument{\edef\std@minuscode{\the\mathcode`-}}
\makeatother

\DeclareMathActive{O}{\operatorname{\std{O}}}

\newtheorem{theorem}{Theorem}
\newtheorem{lemma}[theorem]{Lemma}
\newtheorem{observation}[theorem]{Observation}
\newtheorem{claim}[theorem]{Claim}

\newtheorem{definition}[theorem]{Definition}

\newcommand\bc[1]{\left(#1\right)}
\newcommand\ceil[1]{\left\lceil{#1}\right\rceil}
\newcommand\floor[1]{\left\lfloor{#1}\right\rfloor}

\newcommand{\N}{{\mathbb{N}}}

\newcommand{\vX}{\vec{X}}

\newcommand{\Exp}[1]{{\operatorname{\mathbb{E}}[#1]}}
\newcommand{\ExpCond}[2]{{\operatorname{\mathbb{E}}[#1\mid #2]}}
\newcommand{\Prob}[1]{{\operatorname{\Pr}\left[#1\right]}}
\newcommand{\ProbCond}[2]{{\operatorname{\Pr}[#1\mid #2]}}

\newcommand{\Geom}{{\mathrm{Geom}}}
\newcommand{\NegBin}{{\mathrm{NegBin}}}
\newcommand{\CouponCollector}{{\mathrm{CouponCollector}}}
\newcommand{\OWE}{{\mathrm{\std{O}WE}}}

\usepackage{todonotes}

\newcommand{\cwait}{c_{\mathrm{wait}}} 
\newcommand{\cliveness}{c_{\mathrm{live}}}
\newcommand{\creset}{\cliveness}

\newcommand{\UnsafeSilentProt}{{\normalfont\textsc{Ranking}}}
\newcommand{\Ranking}{\UnsafeSilentProt}
\newcommand{\SpaceEfficientRanking}{{\normalfont\textsc{SpaceEfficientRanking}}}

\newcommand{\StableRanking}{{\normalfont\textsc{StableRanking}}}

\newcommand{\ResetProt}{{\normalfont\textsc{PropagateReset}}}
\newcommand{\PropagateReset}{\ResetProt}
\newcommand{\FastLeaderElect}{{\normalfont\textsc{FastLeaderElection}}}
\newcommand{\TriggerReset}{{\normalfont\textsc{TriggerReset}}}

\def\Phase#1{\operatorname{\normalfont\texttt{phase}}(#1)}
\def\Coin#1{\operatorname{\normalfont\texttt{coin}}(#1)}
\def\LivenessCount#1{\operatorname{\normalfont\texttt{aliveCount}}(#1)}
\def\Rank#1{\operatorname{\normalfont\texttt{rank}}(#1)}

\def\WaitCount#1{\operatorname{\normalfont\texttt{waitCount}}(#1)}
\def\qLE#1{\operatorname{\normalfont\texttt{q}_{\normalfont\texttt{LE}}}(#1)}

\def\ResetCount#1{\operatorname{\normalfont\texttt{resetCount}}(#1)}
\def\DelayCount#1{\operatorname{\normalfont\texttt{delayCount}}(#1)}

\def\Dmax{\ensuremath{D_{\mathsf{max}}}}
\def\Rmax{\ensuremath{R_{\mathsf{max}}}}

\def\LECount#1{\operatorname{\texttt{LECount}}(#1)}

\def\Lmax{\ensuremath{L_{\mathsf{max}}}}

\def\LeaderBit#1{\operatorname{\normalfont\texttt{isLeader}}(#1)}
\def\LeaderDone#1{\operatorname{\normalfont\texttt{leaderDone}}(#1)}
\def\CoinCount#1{\operatorname{\normalfont\texttt{coinCount}}(#1)}

\def\FlipBit#1{\operatorname{\texttt{coin}}(#1)}
\def\LiveCount#1{\operatorname{\texttt{liveCount}}(#1)}
\def\Wmax{\ensuremath{W_{\mathsf{max}}}}

\def\PhaseT#1#2{\operatorname{\normalfont\texttt{phase}}_{#1}(#2)}
\def\CoinT#1#2{\operatorname{\normalfont\texttt{coin}}_{#1}(#2)}

\def\RankT#1#2{\operatorname{\normalfont\texttt{rank}}_{#1}(#2)}

\def\WaitCountT#1#2{\operatorname{\normalfont\texttt{waitCount}}_{#1}(#2)}

\newcommand{\CP}{\ensuremath{C_{\mathrm{Prep}}}}

\newcommand{\CLE}{\ensuremath{C_{\mathrm{LE}}}}
\newcommand{\QLE}{\ensuremath{Q_{\mathrm{LE}}}}
\newcommand{\CT}{\ensuremath{C_{\mathrm{T}}}}
\newcommand{\CM}{\ensuremath{C_{\mathrm{Main}}}}
\newcommand{\QM}{\ensuremath{Q_{\mathrm{Main}}}}
\newcommand{\QRanking}{\ensuremath{Q_{\mathrm{Ranking}}}}

\newcommand{\CS}{\ensuremath{C_{\mathrm{SR}}}}
\newcommand{\CSR}{\CS}
\newcommand{\CSRPlus}{\ensuremath{C_{\mathrm{SR+}}}}
\newcommand{\CL}{\ensuremath{C_{\mathrm{L}}}}
\newcommand{\CWF}[2]{\ensuremath{C_{#1,#2}}}
\newcommand{\CWFWait}[1]{\CWF{#1}{\mathrm{wait}}}
\newcommand{\CWFRank}[1]{\CWF{#1}{\mathrm{rank}}}

\newcommand*{\boldone}{\text{\usefont{U}{bbold}{m}{n}1}}

\def\?#1{}
\def\paragraph#1{\subsubsection*{#1. \?}}

\makeatletter
\def\subsubsection{\@startsection{subsubsection}{3}{0pt}{1.5ex plus 0.2ex minus 0.2ex}%
{0ex}{\normalfont\normalsize\itshape}}%
\makeatother

\allowdisplaybreaks

\title{Silent Self-Stabilizing Ranking:\\Time Optimal and Space Efficient\thanks{
Research was funded by DFG under project number 491453517  and by Austrian Science Fund (FWF) under project number I 5862-N.
}}

\makeatletter
\newcommand{\linebreakand}{%
  \end{@IEEEauthorhalign}
  \hfill\mbox{}\par
  \mbox{}\hfill\begin{@IEEEauthorhalign}
}
\makeatother

\author{
\IEEEauthorblockN{Petra Berenbrink}
\IEEEauthorblockA{\textit{University of Hamburg} \\
Hamburg, Germany \\
petra.berenbrink@uni-hamburg.de}
\and
\IEEEauthorblockN{Robert Elsässer}
\IEEEauthorblockA{\textit{University of Salzburg} \\
Salzburg, Austria \\
robert.elsaesser@plus.ac.at}
\and
\IEEEauthorblockN{Thorsten Götte}
\IEEEauthorblockA{\textit{University of Hamburg} \\
Hamburg, Germany \\
thorsten.goette@uni-hamburg.de}
\linebreakand
\IEEEauthorblockN{Lukas Hintze}
\IEEEauthorblockA{\textit{University of Hamburg} \\
Hamburg, Germany \\
lukas.rasmus.hintze@uni-hamburg.de}
\and
\IEEEauthorblockN{Dominik Kaaser}
\IEEEauthorblockA{\textit{TU Hamburg} \\ Hamburg, Germany \\ dominik.kaaser@tuhh.de}
}

\begin{document}

\maketitle

\begin{abstract}
We present a silent, self-stabilizing ranking protocol for the population protocol model of distributed computing, where agents interact in randomly chosen pairs to solve a common task.
We are given $n$ anonymous agents, and the goal is to assign each agent a unique rank in $\{1, \dots, n\}$.
Given unique ranks, it is straightforward to select a designated leader.
Thus, our protocol is a self-stabilizing leader election protocol as well.

Ranking requires at least $n$ states per agent; hence, the goal is to minimize the additional number of states, called overhead states.
The core of our protocol is a space-efficient but \emph{non-self-stabilizing} ranking protocol that requires only $n + O(\log n)$ states.
Our protocol stabilizes in $O(n^2\log n)$ interactions w.h.p.\ and in expectation, using $n + O(\log^2 n)$ states in total.
Our stabilization time is asymptotically optimal (see \citeauthor{DBLP:conf/podc/BurmanCCDNSX21}, PODC'21).
In comparison to the currently best known ranking protocol by \citeauthor{DBLP:conf/podc/BurmanCCDNSX21}, which requires  $n + \Omega(n)$ states, our result exponentially improves the number of overhead states.
\end{abstract}


\section{Introduction}

The population protocol model \cite{DBLP:journals/dc/AngluinADFP06} is a simple yet expressive computational model for distributed computing. A population of $n$ anonymous \emph{agents} is given and the agents execute a protocol to solve a common task. At any time each agent is on one state out of a given set of states.
In a sequence of discrete time steps pairs of agents are chosen uniformly at random to interact.
In each interaction, the selected agents update their states according to a common transition function.
The required number of states and the number of interactions until a valid configuration is reached form the main performance criteria of population protocols.  
The model has many applications, for example the authors of
\textcite{DBLP:journals/dc/AngluinADFP06} motivate the model with sensor networks, where devices with limited resources are required to perform simple computations. 
Other examples entail chemical reaction networks \cite{DBLP:journals/nc/SoloveichikCWB08} and DNA computing \cite{CDS+13}. Also, certain
 biochemical regulatory processes in living cells can be modeled as population protocols
\cite{CCN12}. 
We refer to the surveys by \textcite{DBLP:journals/sigact/AlistarhG18,DBLP:journals/eatcs/ElsasserR18} for further details and applications.

For many computational tasks in this model we are confronted with the following dilemma: the lack of unique identifiers prevents us from solving many problems efficiently. Unfortunately assigning and maintaining unique identifiers is a notoriously difficult problem, especially if memory (which means in our model the number of states) is scarce, and the systems are prone to faults.
In this paper, we consider the problem of self-stabilizing ranking in the population protocol model.
It is assumed that the agents start in an arbitrary configuration, and the goal is to assign a unique rank from  $\{1, \dots, n\}$ to each agent.
We are interested in
self-stabilizing protocols, which are protocols that are guaranteed to eventually reach a valid configuration. This has to hold for any arbitrary initial configuration, including configurations following transient faults.

The ranking problem is closely related to leader election. Here the goal is to reach a configuration where exactly one agent is in a leader state while all other agents are in so-called follower states.
Given unique ranks it is straightforward to select a leader, e.g., by declaring the agent with rank $1$ to be the leader.
For the ranking problem $n$ is a trivial lower bound on the size of the state space needed by any ranking protocol. 
This holds since each of the indistinguishable agents has to be able to adopt any of the $n$ ranks. 
We, therefore, refer to the states that are required in addition to storing the $n$ ranks as overhead states.


\paragraph{Results in a Nutshell}
We present a novel self-stabilizing protocol for the ranking problem which stabilizes in $O(n^2 \log n)$ interactions w.h.p.
It belongs to the natural class of so-called \emph{silent} protocols which are protocols where, at a certain point, no agent changes its state any longer.
Note our time complexity is optimal within this class \cite{DBLP:conf/podc/BurmanCCDNSX21}.
Our algorithm uses only $O(\log^2 n)$ overhead states in addition to the $n$ states required to store the ranks of the agents.
We underscore the prohibitive nature of this stringent memory restriction: the \emph{additive} overhead of size $O(\log^2 n)$ does not allow agents to hold their rank together with any additional piece of information. Indeed, even a single additional bit would immediately double the state space size.
For example, this rules out that ranked agents participate in a phase-clock to synchronize the protocol, or store the information whether they are a leader or not.
Similarly, if a leader exits which is then used to assign the ranks (as in, e.g., \cite{DBLP:conf/podc/BurmanCCDNSX21,DBLP:conf/opodis/GasieniecJLL21}), the leader agent cannot store any information about the ranks assigned so far, including its own rank.
In order to overcome this issue, our protocol still works with a leader, but our leader is blissfully unaware of its special state.
This piece of information is only revealed to the ``unaware'' leader once it communicates with an unranked agent $u$: depending on the combination of its own state and the state of $u$, the unaware leader realizes which role it has,  which allows the agent to assign the correct rank to $u$. Finally, after ranking,  our protocol solves the problem of self-stabilizing leader election by selecting the agent with the lowest rank as the leader.


\section{Related Work}

\paragraph{Ranking Protocols}
%

The ranking problem occurs in different communication models and under various assumptions  \cite{DBLP:conf/wdag/AlistarhAGGG10,DBLP:conf/podc/AlistarhACGZ11,DBLP:journals/dc/BrodskyEW11,DBLP:conf/podc/GiakkoupisW12,DBLP:journals/jpdc/BerenbrinkBEFN21}.
In this overview we focus on results in the population protocol model.
Ranking in population protocols typically considers so-called \emph{safe} and \emph{silent} protocols.
In a safe protocol, once an agent receives a rank, this rank is never changed.
In a silent protocol the population eventually reaches a final configuration in which agents no longer \mbox{change their states.}

\Textcite{DBLP:conf/opodis/BeauquierBRR12} present a population protocol for a generalization of the ranking problem.
Their protocol distributes $m$ unique labels with $m \geq n$ among the agents.
For $m=n$, this corresponds to the ranking problem.
The authors focus on the feasibility of the solution but do not analyze the time needed for the population to stabilize.
Another set of self-stabilizing ranking algorithms is provided by \textcite{burman_et_al:LIPIcs.DISC.2019.9}.
Again, their focus lies on feasibility under weak scheduler while additionally optimizing the required number of states.

\Textcite{DBLP:conf/opodis/GasieniecJLL21} present two safe and silent ranking protocols, one for a range of $[1,(1+\epsilon) \cdot n]$ and one for the optimal range $[1,n]$.
The first protocol requires $O(n\log(n)/\epsilon)$ interactions w.h.p.\ and uses $(2+\epsilon) n + O(n^{\alpha})$ states for an arbitrary constant $\alpha$.
While the protocol used $\Omega(n)$ overhead states, for $\epsilon=\Omega(1)$ the protocol only requires an asymptotically optimal number of $O(n \log n)$ interactions.
In addition to the upper bound, the authors show for this protocol a lower bound, which can also be generalized to a wider class of protocols:
to assign ranks from the range $[1, n + r]$ the expected number of interactions is at least $n\cdot (n-1)/(2(r+1))$.
For the optimal range $[1, n]$, the authors present a protocol which requires $O(n^3)$ interactions in expectation and $n+5\sqrt{n}+O(n^c)$ states, where $c$ can be an arbitrarily small constant.
A generalization of the protocol is parameterized by $\epsilon\ge n^{-1}$ and uses $(1+9\sqrt{\epsilon})\cdot n+O(\log\log n)$ states and $O(n^2/\epsilon)$ interactions w.h.p.
Finally, the authors show a general lower bound for safe and silent protocols that even holds if a designated leader agent is present from the beginning:
any safe and silent protocol that uses a range of $[1,n]$ and produces a valid ranking with probability larger than $1 - 1/n$ requires at least $n +\sqrt{n-1}-1$ states (see the full version \cite{gasieniec2021efficient} of \cite{DBLP:conf/opodis/GasieniecJLL21}).
Note that none of these protocols are self-stabilizing.

Very recently, \textcite{gąsieniec2024} present a different approach for deanonymizing a population that is orthogonal to the previously mentioned approaches.
Here, all agents agree on a common coordinate system and assign themselves unique points in that system.

\paragraph{Leader Election}
Leader election is a prominent problem in the population protocol model that is closely related to ranking.
A long series of papers on non-self-stabilizing leader election \cite{DBLP:conf/icalp/AlistarhG15,DBLP:conf/podc/BilkeCER17,DBLP:conf/soda/AlistarhAG18,DBLP:conf/soda/BerenbrinkKKO18,DBLP:conf/soda/AlistarhAEGR17,DBLP:conf/soda/GasieniecS18,DBLP:journals/tpds/SudoOIKM20,DBLP:conf/spaa/GasieniecSU19} has lead to the currently best known protocol by \textcite{DBLP:conf/stoc/BerenbrinkGK20} which uses $O(\log \log n)$ states and stabilizes in asymptotically optimal $O(n\cdot \log n)$ interactions in expectation. 
\Textcite{DBLP:conf/wdag/SudoEIM21} present a so-called loosely-stabilizing protocol that elects a leader starting from any arbitrary initial configuration.
Informally, a loosely-stabilizing protocol converges quickly from an arbitrary initial configuration to a valid configuration with a unique leader and then remains in a valid configuration for a long time.
Their protocol improves on earlier works on loosely-stabilizing leader election \cite{DBLP:journals/tcs/SudoNYOKM12,DBLP:journals/tcs/SudoOKMDL20}.

To the best of our knowledge, all efficient self-stabilizing leader election protocols in the population protocol model are based on self-stabilizing ranking protocols.
These protocols compute a ranking which trivially implies a leader.
In this setting, \textcite{DBLP:journals/mst/CaiIW12} present a silent self-stabilizing leader election protocol that requires $O(n^3)$ interactions w.h.p.\ using $n$ states.
Furthermore, \textcite{DBLP:conf/podc/BurmanCCDNSX21} present three self-stabilizing protocols for leader election based on ranking. 
The first one is silent and requires $O( n^2\log n)$ interactions w.h.p.\ using $O(n)$ states.
The second protocol and the third protocol are both non-silent.
The second one needs only $O(n \log n)$ interactions w.h.p.\ at the expense of an exponential number of $\exp(O(n^{\log n}\cdot \log n))$ states.
The third one allows a trade-off between the number of states and the running time controlled by parameter $1 \leq H = O(\log n)$.
It requires $O(H n^{1+1/(H+1)})$ interactions w.h.p.\ and $O(n^{\Theta(n^H)}\cdot \log n)$ states.
Finally, note that \textcite{DBLP:journals/mst/CaiIW12} show that any self-stabilizing leader election protocol requires at least $n$ states and 
\textcite{DBLP:conf/podc/BurmanCCDNSX21} show that every silent leader election protocol requires $\Omega(n^2)$ interactions in expectation (and $\Omega(n^2 \log n)$ interactions w.h.p.).
Thus, the silent self-stabilizing leader election protocol implied by our ranking protocol matches the lower bound on the time complexity from \cite{DBLP:conf/podc/BurmanCCDNSX21} and almost matches the state complexity from \textcite{DBLP:journals/mst/CaiIW12}, except for the \emph{additive} $O(\log^2 n)$ states.

\paragraph{Ranking of Anonymous Networks}
Another related problem is assigning a rank to all nodes of an anonymous network.
The network is modeled by a connected graph $G = (V, E)$ whose edges may change over time.
Time proceeds in synchronous rounds, and in each round, a node $v \in V$ can send a message to all its neighbors.
The nodes have no identifiers, but usually, they differentiate between their neighbors based on \emph{port numbers}.
\Textcite{DBLP:conf/icdcs/KowalskiM21} present a (non-self-stabilizing) leader election protocol that runs in $O(t_{\mathrm{mix}} \log^2 n)$ time w.h.p., where $t_{\mathrm{mix}}$ is the mixing time of simple random walk on $G$.
\Textcite{DBLP:conf/focs/LunaV22, DBLP:conf/mfcs/LunaV24} consider the ``reverse'' problem of determining the number of nodes $n$ given a predetermined leader.
Their algorithm takes $O(n)$ rounds.
Note that algorithms for population protocols can usually be transferred to anonymous networks:
the authors of \cite{DBLP:conf/opodis/AlistarhGR21} present a general framework that simulates a population protocol on a graph. Their approach is random-walk based and the runtime depends on the properties of the graph $G$ such as diameter or conductance. The converse direction is not straightforward:
in an anonymous network \emph{all} nodes can communicate with \emph{all} their neighbors in each time step, while in a population protocol  \emph{exactly one pair} of agents interacts per time step.

\section{Model and Results}
\label{sec:model_results}


We consider a set $V$ of $n$ agents. Each agent $v$ has a state $x(v)$ from a state space $Q$.
A \emph{configuration} is a vector $(x(v))_{v \in V} \in Q^n$ that maps each agent $v \in V$ to its state $x(v)$.
Time is measured in discrete steps.
In each time step, two agents are chosen uniformly at random to interact.
The two chosen agents update their states according to a common transition function.
The configuration $\vX_0$ at time $0$ is called \emph{initial configuration}.
Due to the random interactions, the configuration at time $t>0$ is a random vector $\vX_t \in Q^n$.
Since we consider self-stabilization, the initial configuration is arbitrary.
The goal is to reach a \emph{valid configuration} in which all agents have a unique rank from $[n]$ defined as $[n] = \{1,\ldots, n\}$.
A valid configuration is called \emph{stable} if no sequence of interactions exists that changes the output on any agent.
We say that a protocol stabilizes after $\tau$ steps if $\vX_\tau$ is valid and stable.

A population protocol with state space $Q$ is \emph{self-stabilizing} with respect to a set of configurations $C_L \subset Q^n$ if and only if it fulfills the following two properties.
\begin{itemize}
    \item \emph{Closure:}\quad If $\vX_t \in C_L$ for some $t$, then $\vX_{t+1} \in C_L$.
    If additionally $\vX_{t+1} = \vX_{t}$, i.e., no agent changes its state, the protocol is \emph{silent}.
    \item \emph{Probabilistic Stabilization:}\quad For every $\vX_t \in Q^n$ we have $\displaystyle\lim_{\tau \to \infty} \Prob{\vX_{t+\tau} \in C_L}=1$.
    Note that in contrast to other models, we cannot guarantee deterministic stabilization for population protocols due to the random interactions.
\end{itemize}
For our protocols, we let $C_L$ be the set of all permutations of~$[n]$.
That is, $C_L$ is the set of all configurations in which every agent is assigned a rank, and all ranks are unique.
Together with an output function that maps a rank of $1$ to ``leader'' and any other state to ``follower'' this immediately carries over to self-stabilizing leader election.

The following theorem is our first main result. 
The corresponding protocol \SpaceEfficientRanking\ is introduced and analyzed in \cref{sec:unstable_ranking}.

\begin{theorem}[label=thm:non-self-stab-ranking,restate=thmOne]
    \SpaceEfficientRanking\ is a silent population protocol with $n + \Theta(\log n)$ states that reaches a valid ranking in $O(n^2 \log n)$ interactions w.h.p.
\end{theorem}
In our second main result, we transform the protocol from \cref{thm:non-self-stab-ranking} into a self-stabilizing protocol.
Most notably, we amend it with error-detection and a resetting mechanism.
We present the required changes and analyze the corresponding protocol \Ranking+\ in \cref{sec:ss_protocol}.


\begin{theorem}[label=thm:ss-ranking,restate=thmTwo]
    \Ranking+\ is a silent population protocol for self-stabilizing ranking that requires $n + O(\log^2 n)$ states and stabilizes in $O(n^2 \cdot \log n)$ interactions w.h.p.
\end{theorem}

\section{Non-self-stabilizing Ranking}
\label{sec:unstable_ranking}

In this section we describe our non-self-stabilizing ranking protocol \SpaceEfficientRanking.
Intuitively, the protocol works as follows.
All agents start with leader election using the protocol from \cite{DBLP:conf/soda/GasieniecS18}.
As soon as one agent is elected as the unique leader, this agent starts the actual ranking.
The ranking then runs in multiple phases, and in each phase a contiguous interval of ranks is assigned: in phase $1$, the leader assigns ranks $n/2+1$ to $n$, in phase $2$ the leader assigns ranks $n/4+1$ to $n/2$, and so on (assuming $n$ is a power of $2$ for now).

Recall that one agent cannot remember all necessary information such as being the unique leader or not, the current phase number, the next rank to be assigned, and its rank at the same time.
Instead, we distribute this information across multiple agents:
agents either store a rank or the current phase index, and nothing else.
In phase $k$, the leader stores a rank between $1$ and $n/2^{k+1}$.

Now suppose that we are in phase $k$ where the ranks $n/2^{k+1}+1$ to $n/2^k$ are assigned.
Our protocol ensures that the leader is the sole agent with a rank  $r \leq n/2^{k+1}$.
When the leader interacts with an unranked agent, it assigns rank $n/2^{k+1} + r$.
If $r$ is below $n/2^{k+1}$, the leader increments its own rank by $1$.
Otherwise, it starts a broadcast that increases the phase to $k+1$.
The leader goes into a special waiting state until the broadcast has finished. 
Then, it assigns itself rank $r=1$ again and thus starts the next phase $k+1$.

Note that the leader is ``unaware'' of its special state.
Only when it interacts with an unranked agent, it realizes that it is indeed the leader. 
Our protocol ensures that at all times there is w.h.p.\ only one \emph{unaware leader}, namely the one elected in the beginning.

\medskip

Before we give the formal protocol definition, we first give a formal overview of the state space.
We assume that the exact value of $n$ is known. This is in fact necessary for leader election and thus also for ranking, see Theorem 1 in \cite{DBLP:journals/mst/CaiIW12}.
\begin{align*}
    \QRanking &= \underbrace{\QLE}_{\texttt{q}_{\texttt{LE}}} \times \underbrace{\{0, 1\}}_{\mathclap{\texttt{leaderDone}}} \uplus \underbrace{\left\{  1, \ldots,  \ceil{\cwait\cdot\log n}\right\}}_{\texttt{waitCount}}\\
    & \phantom{{}={}}\uplus \underbrace{\left\{  1, \ldots,  \lceil\log n\rceil\right\}}_{\texttt{phase}} \uplus \underbrace{\left\{  1, \ldots,  n\right\}}_{\texttt{rank}}.
\end{align*}
Here, $\uplus$ is the disjoint union of two sets and $\QLE$ is the state space of the leader election protocol by \textcite{DBLP:conf/soda/GasieniecS18}.
The expression $\qLE{v} \in \QLE$ contains the leader-election state of $v$
    and is initialized to the initial state $q_0 \in \QLE$.
Similarly to \cite{DBLP:conf/podc/BerenbrinkKR19}, we assume that the leader election protocol provides, additionally to the state $\qLE{v} \in \QLE$, a variable $\LeaderDone{v}$ which is set to $1$ when agent $v$ assumes that the leader election is done.
When all agents have $\LeaderDone{v}=1$ there is w.h.p.\ exactly one leader agent~$\ell$.
The variables $\WaitCount v$ and $\Phase v$ are both used to guide our ranking protocol and $\Rank v$ is used to store the rank of agent $v$.
These values are all initialized with $\bot$, indicating that the value is as-yet undefined.
Throughout the run of the protocol,
    each agent can have exactly one value of $\qLE{v}$, $\WaitCount{v}$, $\Phase{v}$, or $\Rank{v}$ be \emph{not} equal to $\bot$, and $\LeaderDone{v} \neq \bot$ if and only if $\qLE{v} \neq \bot$.
We call these agents leader-electing agents, waiting agents, phase agents, and ranked agents, respectively.


\medskip

\noindent Our protocol consists of two parts,  $\SpaceEfficientRanking$ (\cref{alg:unsafe_silent_plus}) and $\UnsafeSilentProt$ (\cref{alg:unsafe_silent}).
The former elects a unique leader and transitions to the latter, which assigns the ranks given a leader.

\begin{algorithm}[b]{$\SpaceEfficientRanking(u, v)$%
\label{alg:unsafe_silent_plus}}
${\textwidth=\columnwidth\sepline{Leader Election}}$
if $\qLE u, \qLE v \neq \bot$ then /*two leader-electing agents interact\label{ln:le-s}*/
	execute $\textsc{ElectLeader}(u, v)\label{ln:le-e}$

${\textwidth=\columnwidth\sepline{The Leader is Done}}$
if $\exists\ell \in \{u, v\}:$ ($\LeaderBit \ell = \LeaderDone \ell = 1$) then$\label{ln:le-s-s}$
    $(\qLE \ell, \LeaderDone \ell) \gets (\bot, \bot)$ /*$\ell$ forgets LE state*/
    $\WaitCount \ell \gets \ceil{\cwait \cdot \log n}$ /*$\ell$ becomes a waiting agent\label{ln:le-s-e}*/
    return$\label{ln:le-s-e-return}$

${\textwidth=\columnwidth\sepline{Propagate Start of Ranking}\iffalse/*LE agent $w$ meets non-LE agent $x$*/\fi}$
if $\qLE w \neq \bot$ and $\qLE x = \bot$ for a $\{w, x\} = \{u, v\}$ then$\label{ln:oe-s}$
    $(\qLE w, \LeaderDone w) \gets (\bot, \bot)$ /*$w$ forgets its LE state*/
    $\Phase w \gets 1$ /*$w$ becomes a phase agent\label{ln:oe-e}*/

${\textwidth=\columnwidth\sepline{Ranking Protocol}}$
if $\qLE u, \qLE v = \bot$ then /*non-leader-electing agents interact*/
    execute $\UnsafeSilentProt(u, v)$
\end{algorithm}

$\SpaceEfficientRanking$ begins by electing a unique leader using the protocol from \cite{DBLP:conf/soda/GasieniecS18}:
all agents start in a leader-election state.
Whenever two leader-electing agents interact, they follow the transition function of the leader election protocol (lines \ref{ln:le-s}--\ref{ln:le-e}).
The real ranking protocol is started by the leader $\ell$ as soon as $\LeaderDone{\ell}$ is set to $1$.
Then, $\ell$ immediately forgets its leader-election state (setting it to $\bot$) and sets $\WaitCount \ell = \ceil{\cwait \log n}$ (lines \ref{ln:le-s-s}--\ref{ln:le-s-e-return}).
This initiates a one-way epidemic informing all agents that Phase~1 starts (lines \ref{ln:oe-s}--\ref{ln:oe-e}).
At this time the agents are w.h.p.\ in a configuration with the following property (see \cref{leader}):
a unique leader agent~$\ell$ has $\WaitCount \ell = \ceil{\cwait \log n}$, and all other agents are either in a state of $\QLE$ where $\LeaderBit v = 0$, or have $\Phase u = 1$.
We call this set of configurations $\CS$ for \underline{s}tart \underline{r}anking.

\begin{figure}
\begin{algorithm}[H]{\UnsafeSilentProt($u$, $v$)\label{alg:unsafe_silent}%:
}
if $\Phase v = \bot$ then return /*if $v$ has rank, do nothing*/

if $\Rank u \neq \bot$ then
    let $k = \Phase{v}$

    if $1 \leq \Rank{u} \leq f_k - f_{k+1}$ then /*$u$ may assign a rank to $v$\label{ln:rank-s}*/
        $(\Phase v, \Rank v) \gets (\bot, f_{k+1} + \Rank{u})$
        if $\Rank{u} < f_k - f_{k+1}$ then /*phase not done*/
            $\Rank u \gets \Rank u + 1\label{ln:rank-se}$ 
        else if $k < \lceil \log_2 (k) \rceil\label{ln:rank-e-pre}$ /*$u$ reached end of non-final phase*/
            $(\Rank u, \WaitCount u) \gets (\bot, \ceil{\cwait \cdot \log n})\label{ln:rank-e}$

    if $\Rank u = f_k$ then /*$u$ has last rank in phase $k$\label{ln:transition-s}*/
        $\Phase{v} \gets \Phase{v} + 1\label{ln:transition-e}$

if $\Phase u \neq \bot$ then$\label{ln:phase-oe-s}$
/*if both agents are phase agents, broadcast more advanced phase*/
    $\Phase u, \Phase v \gets \max\{\Phase u, \Phase v\}\label{ln:phase-oe-e}$

if $\WaitCount u \neq \bot$ then$\label{ln:dec-waitcount-s}$
/*decrement wait counter and ultimately transition to rank $1$*/
    $\WaitCount u \gets \WaitCount u - 1\label{ln:dec-waitcount-e}$
    if $\WaitCount u = 0$ then$\label{ln:waitcount-0-s}$
        $\WaitCount u, \Rank u \gets \bot, 1\label{ln:waitcount-0-e}$
\end{algorithm}

\end{figure}



Having reached a configuration in $\CS$, the actual ranking protocol $\UnsafeSilentProt$ described in \cref{alg:unsafe_silent} takes over.
Whenever the agent $\ell$ with $\WaitCount \ell \neq \bot$ interacts with a phase agent, it decrements $\WaitCount{\ell}$  (lines \ref{ln:dec-waitcount-s}--\ref{ln:dec-waitcount-e}).
As soon as this counter reaches zero, $\ell$ assigns itself rank $1$, thereby taking on the role of unaware leader (lines \ref{ln:waitcount-0-s}--\ref{ln:waitcount-0-e}).
The now-unaware leader $\ell$ with $\Rank {\ell} = 1$ starts the ranking.
At that point each agent $v \neq \ell$ has $\Phase v = 1$ w.h.p.

The ranking is performed in $\log_2(n)$ phases as sketched above.
At the beginning of phase $k$, the leader $\ell$ has $\Rank \ell = 1$, and each agent~$v$ that is not yet ranked has $\Phase v = k$.
Writing $f_k$ for the maximal rank assigned in phase $k$, we let $f_1 = n$ and  $f_i = \lceil f_{i-1} / 2 \rceil$ for all $i > 1$.
Note that if $n$ is a power of two, then $f_k = n / 2^{k-1}$.
In general, in the $k$th phase, ranks $f_{k+1}+1,\ldots, f_k$ are assigned.

Let us consider an interaction between $\ell$ and an unranked agent $v$ with $\Phase v = k$.
If $\Rank {\ell}=r \leq f_k - f_{k+1}$,
    then agent~$v$ sets $\Rank v=f_{k+1} + r$,
    and as long as $r < f_{k+1} - f_k$, agent~$\ell$ increments $\Rank{\ell}$ by $1$ (lines \ref{ln:rank-s}--\ref{ln:rank-se}).
Otherwise, if $r = f_k - f_{k+1}$, then $v$ received the largest rank $f_k$ of phase $k$.
If $k=\lceil \log_2 n \rceil$, this was the final phase.
$\ell$ remains with $\Rank \ell = 1$, and the protocol is silent from now on.
Otherwise, a phase transition starts:
$\ell$ forgets its rank ($\Rank \ell = \bot$) and sets $\WaitCount \ell = \ceil{\cwait \log n}$ (lines \ref{ln:rank-e-pre}--\ref{ln:rank-e}).
$\WaitCount \ell$ is decremented whenever $\ell$ meets a phase agent (line \ref{ln:dec-waitcount-e}).
When an agent $v$ with $\Phase v = k$ meets the agent with rank $f_k$, it can safely infer that phase $k$ is finished.
Thus, it increments its phase (see lines \ref{ln:transition-s}--\ref{ln:transition-e}), and the incremented phase spreads via one-way epidemic among phase agents (see lines \ref{ln:phase-oe-s}--\ref{ln:phase-oe-e}).
The leader's and phase agents' transitions are timed such that when $\WaitCount \ell = 0$, all phase agents have updated their phase w.h.p., and the leader can safely set $\Rank \ell = 1$ (lines \ref{ln:waitcount-0-s}--\ref{ln:waitcount-0-e}).


\subsection{Analysis}
\label{sec:unstable_ranking_analysis}
In this section, we show \cref{thm:non-self-stab-ranking}.
First, we calculate the number of states used by \cref{alg:unsafe_silent_plus}.
Since these states of ranked agents, waiting agents, phase agents, and leader-electing agents are disjoint,
    the protocol uses $n + \ceil{\cwait \cdot \log n} + \ceil{\log n} + 2\lvert\QLE\rvert = n + \Theta(\log n)$ states (as $\lvert{\QLE}\rvert = O(\log\log n)$),
    as claimed.
It remains to show the correctness of the protocol and to calculate its runtime. The fact that the protocol is silent follows directly from the definition of the protocol.

The proof is split into two parts. First we show in
\cref{leader} that after $O(n\log n)$ interactions a state in $\CS$ (configurations in which the actual ranking is started) is reached.
Then we show in \cref{lemmacorrectness_non_ss} that in another $O(n^2 \log n)$ interactions each agent receives a unique rank. \Cref{thm:non-self-stab-ranking} then follows directly from \cref{leader,lemmacorrectness_non_ss}.

\begin{lemma}\label{leader}
W.h.p., there is a $\tau = O(n \log^2 n)$ such that  $\vX_{t+\tau} \in \CS$.
\end{lemma}

\cref{leader} is mostly a direct consequence of the correctness of the leader election protocol. 

\begin{lemma}
\label{lemmacorrectness_non_ss}
Let $c$ be a sufficiently large constant, and assume $\vX_t \in \CS$ and $\cwait \geq 24 + 48 \gamma$.
Then, there is a $\tau \in [c \cdot n^2 \log n]$ such that $\vX_{t+\tau}$ is a configuration in $\CL$.
\end{lemma}

The main idea of the proof is to show that the protocol alternates between so-called initial waiting and initial ranking configurations ($\CWFWait{k}$ and $\CWFRank{k}$, defined below in \cref{def:initial_configurations}).
The proof of the lemma is divided into two cases.
In \cref{lem:unsafe_waiting_subphase} we show that, starting with an initial waiting configuration in $\CWFWait{k}$, the protocol is in an initial ranking configuration after  $O(2^k n \log n)$ interactions.
Also, starting with an initial ranking configuration $\in \CWFRank{k}$, the protocol is in an initial waiting configuration after $O(n^2 + 2^k n\log n)$ interactions (\cref{lem:unsafe_ranking_subphase}).
From this \cref{lemmacorrectness_non_ss} follows via induction over the number of phases. Note that the number of interactions in a waiting phase increases with $k$.
This is because the one-way epidemics are restricted to the unranked agents.
We prove \cref{lem:unsafe_waiting_subphase,lem:unsafe_ranking_subphase} in the remainder of this section.
To state them, 
    we need the following technical definition, with the largest rank in phase~$k$, $f_k$, defined as above as $f_1 = n$ and $f_i = \lceil f_{i-1} / 2 \rceil$ for $i > 1$.

\begin{definition}\label{def:initial_configurations}
    Let $k \in [\ceil{\log_2 n}]$ be a phase index.
    \begin{enumerate}
        \item The set of \emph{initial waiting configurations} for phase $k=1$, called $\CWFWait{1}$ is $\CS$.
        \item The set of \emph{initial waiting configurations} for phase $k>1$, called $\CWFWait{k}$, is the set of configurations with
        \begin{enumerate}
            \item a unique waiting agent $\ell$, which has $\WaitCount \ell = \ceil{\cwait \cdot \log n}$.
            \item for each $i \in [f_k + 1, n]$, a unique agent $u_i$ with $\Rank{u_i} = i$, and these are the only ranked agents,
            \item $\Phase {w} \leq k$ for all phase agents,
            \item no leader-electing agents.
        \end{enumerate}
        \item The set of \emph{initial ranking configurations} called $\CWFRank{k}$ for phase $k$
            is defined as the set of configurations with \begin{enumerate}
                \item a unique unaware leader $\ell$ having $\Rank{\ell} = 1$,
                \item for each $i \in [f_{k-1} + 1, n]$ a unique agent $u_i$ having $\Rank{u_i} = i$, and these and the unaware leader are the only ranked agents,
                \item $\Phase{w} = k$ for all phase agents,
                \item no leader-electing or waiting agents.
            \end{enumerate}
    \end{enumerate}
\end{definition}

\begin{lemma}\label{lem:unsafe_waiting_subphase}
    For any $k$ with $1 \leq k \leq \ceil{\log_2 n}$ and any $\gamma > 0$ the following statement holds.
    Assume that $\vX_t$ is an arbitrary configuration in $\CWFWait{k}$ and $\cwait \geq 24 + 48 \gamma$.
    Then with probability of at least $1 - 5 n^{-\gamma}$ there is a $\tau \leq \bc{\cwait + \gamma} 2^k \cdot n \log n$ such that $\vX_{t+\tau} \in \CWFRank{k}$.
\end{lemma}

\newcommand{\Twait}{{T_{\mathrm{wait}}}}
\newcommand{\TOWE}{{T_{\mathrm{\std{O}WE}}}}

\begin{proof}

    \emph{Case $k > 1$:}\quad
    Let $\ell$ be the unique waiting agent at the time $t$ where $\vX_t \in \CWFWait{k}$,
        and let $\Twait$ be defined such that $t + \Twait$ is the first time after $t$ at which $\ell$ becomes ranked again.
    At all times in the interval $[t, t+\Twait]$
        there are $n - 1 - (n-f_k) = f_k - 1$ phase agents. This holds
        since there are $n - f_k$ ranked agents and one waiting agent (agent $\ell$) in the population, the other agents are phase agents.
    Since $\WaitCount \ell$ is decremented every time $\ell$ meets a phase agent,
        after
        $\ceil{\cwait \log n}$ such meetings $\ell$ becomes ranked again.
    Thus, $\Twait$ has the negative binomial distribution $\NegBin\bc{\ceil{\cwait \log n}, \frac{f_k - 1}{n(n-1)}}$.
    From the upper and lower tail bounds (see \iflong\cref{lem:negbin_tailbound} in \cref{apx:tail-bounds}\else full version\fi)
    we get
    \begin{align}
    \label{eq_lemma6_1}
    \Prob{\Twait \leq \frac{n^2}{f_k - 1} \cdot \bc{\cwait + \gamma} \log n} & \geq 1 - n^{-\gamma}
    ~ \text{and} \\
    \Prob{T_{\mathrm{wait}} > \frac{1}{4} \cdot \frac{n(n-1)}{f_k - 1} \cdot \cwait \log n} & \geq 1 - n^{-\cwait/6}.
    \label{eq_lemma6}
    \end{align}
    Taking into account that $f_k \geq 2$ for all $k$,
        we have $1/(f_k - 1) \leq 2/f_k \leq 2^k / n$,
        which together with the upper bound on $\Twait$ above yields
        $\Twait \leq  \bc{\cwait + \gamma} \cdot 2^k \cdot n \log n$.

    Next we prove that with probability at least $1 - 3n^{-\gamma}$,
        at time $t + \Twait$
    all phase agents $w$ have $\Phase w =k$
        (remember, no phase agent can become ranked $t + \Twait$).
    Each phase agent $w$ switches $\Phase w$ to $k$ by being prompted via a one-way epidemic spread among the phase agents. The initiator of this epidemic is the ranked agent $v$ with $\Rank v = f_k$.

    Consider now a modified protocol in which, as long as the epidemic did not reach all phase agents, $\ell$ is not allowed to assign any rank to the agents. Let
        $\TOWE$ be defined such that $t + \TOWE$ is the first time after $t$ at which the one-way epidemic
        reaches all phase agents.
    Clearly, if $\TOWE < \Twait$, then the original and the modified protocol behave identically.
For the modified protocol the upper tail bound \iflong (see \cref{lem:subset_broadcast_upper_tail} in \cref{apx:tail-bounds} with $m = f_k$)\else (see full version) \fi gives
\begin{align}
        \Prob{\TOWE > 3 \frac{n^2}{f_k} \cdot \bc{\log(f_k) + 2\gamma \log n}} \leq 2n^{-\gamma}.
        \label{eq_lemma6_3}
        \end{align}
        Assuming $n \geq 2$ the above bound implies $\TOWE \leq 6 \cdot \frac{n(n-1)}{f_k - 1} \cdot \bc{2\gamma + 1} \log n$.
    The assumption $\cwait \geq 24 + 48 \gamma$
        ensures that this upper bound is not larger than the lower bound on $\Twait$ given in \cref{eq_lemma6}. Note that $\cwait$ is also large enough to ensure $\cwait/6 \geq \gamma$ and by union bound \(\Prob{\TOWE \leq \Twait} \geq 1 - 3n^{-\gamma},\)
        which concludes the proof of the case $k>1$.

\emph{Case $k = 1$:}\quad
For the analysis, we consider the following modified protocol. If two or more leaders are elected by the leader election protocol (lines \ref{ln:le-s}--\ref{ln:le-e} in \cref{alg:unsafe_silent_plus}), the agent $u$ which sets $\LeaderDone{u}$ to $1$ first, remains the leader, while any other elected leader $v$ loses its leader role instead of setting $\LeaderDone{v}$ to $1$. 
Clearly, if in the original protocol one leader is elected, then the two protocols behave identically. 
However, in the modified protocol we enforce that in fact at most one leader is elected. 
Apart from the modification described above, the two protocols have the same transition function.

Let $\ell$ be the unique leader elected in the modified protocol. Then, $\ell$ starts a one-way epidemic, which is spread among the whole population. Similarly to the case $k>1$, define $\Twait$ such that $t + \Twait$ is the first time at which $\ell$ becomes ranked. Also, let $\TOWE$ be defined such that $t + \TOWE$ is the first time at which the one-way epidemic reaches all agents.
We know that the waiting agent $\ell$ decrements its wait counter when interacting with any of the other $n-1 = f_k - 1$ agents. Thus, also in the case $k=1$, we obtain the same upper and lower tail bounds on $\Twait$ as in \cref{eq_lemma6_1,eq_lemma6}. For
$\TOWE$ we obtain the same bound as in \cref{eq_lemma6_3}. Thus, in the modified protocol at time $t+\Twait$ with $\Twait \leq (\cwait+\gamma) \cdot n \log n$ we have $\vX_{t+\Twait} \in \CWFRank{1}$ with probability $1-4 n^{-\gamma}$. Since the original protocol elects a single leader with probability at least $1-n^{-\gamma}$ applying union bound concludes the proof.
\end{proof}

\begin{lemma}\label{lem:unsafe_ranking_subphase}
    For any $k$ with $2 \leq k \leq \ceil{\log_2 n}$ and $\gamma > 0$ the following statement holds.
    Assume that $\vX_t$ is an arbitrary configuration in $\CWFRank{k}$.
    Then with probability at least $1 - n^{-\gamma}$ there is a $\tau \leq 2n^2 + 2\gamma 2^k n \log n$ such that $\vX_{t+\tau} \in \CWFWait{k+ 1}$ when $k < \ceil{\log_2 n}$,
        or $\vX_{t+\tau} \in \CL$ when $k = \ceil{\log_2 n}$.
\end{lemma}

\begin{proof}
    For $1 \leq i \leq f_k - f_{k+1}$, let $C_{k,i}$ be the set of configurations with\begin{enumerate}
        \item a unique unaware leader $\ell$, which has $\Rank \ell = i$,
        \item for each $i \in [f_{k+1} + 1, f_{k+1} + i - 1] \cup [f_k + 1, n]$ a unique agent $u_i$ with $\Rank {u_i} = i$, and these and the unaware leader are the only ranked agents,
        \item $\Phase k$ for all phase agents, and
        \item no leader-electing agents
    \end{enumerate}
    By definition of the protocol,
        in any of these configurations,
        the only kind of interaction which will change the configuration
        is that between the unique unaware leader
            and one of the phase agents (lines \ref{ln:rank-s}--\ref{ln:rank-e} in \cref{alg:unsafe_silent}).
    Assume for now that $\vX_{t'}$ is an arbitrary configuration in $C_{k,i}$.
    Then if $1 \leq i < f_k - f_{k+1}$,
        the next configuration not in $C_{k,i}$ is in $C_{k,i+1}$.
    If $i = f_k - f_{k+1}$ and $k < \ceil{\log_2 n}$,
        the next configuration not in $C_{k,i}$ is in $\CWFWait{k+1}$;
        and if $i = f_k - f_{k+1}$ and $k = \ceil{\log_2 n}$,
        the next configuration not in $C_{k,i}$ is a valid ranking, i.e., in $\CL$.

    Now in $C_{k,i}$ there is one unaware leader
        and $n - 1 - (n-f_k) - (i-1) = f_k - i$ phase agents,
        so in total $f_k - i$ out of $n(n-1)$ possible ordered interaction pairs will lead to the next class of configurations.
        Let $T_{k,i}$ denote the number of interactions between the first time step at which the configuration is
        in $C_{k,i}$ and the first time step at which the configuration leaves $C_{k,i}$.
    Then, $T_{k,i}$
        has distribution $\Geom\bc{\frac{f_k - i}{n(n-1)}}$,
        and the $T_{k,i}$ are independent as the corresponding time steps are disjoint.

In the following analysis we write
        $X \preceq Y$ when $X$ is stochastically dominated by $Y$.
    The total time to reach $\CWFWait{k+1}$ if $k < \ceil{\log_2 n}$, or $\CL$ if $k = \ceil{\log_2 n}$)
        from $\CWFRank{k}$ is the sum of independent geometric random variables.
Since $\Geom(p) \preceq \Geom(q)$ for $p \geq q$
        and by definition of the negative binomial distribution,
        for all $k$,
        \begin{align*}\sum_{i=1}^{f_k - f_{k+1}} T_{k,i} &\sim \sum_{i=1}^{\mathclap{f_k - f_{k+1}}} \Geom\bc{\frac{f_k - i}{n(n-1)}} \preceq \sum_{i=1}^{\mathclap{f_k - f_{k+1}}} \Geom\bc{\frac{f_{k+1}}{n^2}}\\
        & \sim \NegBin\bc{f_k - f_{k+1}, \frac{f_{k+1}}{n^2}},\end{align*}
        with the geometric random variables in the sums being independent.
    By stochastic domination
        and \iflong the \else a \fi tail bound on negative binomial random variables \iflong(see \cref{lem:negbin_tailbound} in \cref{apx:tail-bounds})\fi 
        we thus have
        \[\Prob{\sum_{i=1}^{f_k - f_{k+1}} T_{k,i} \leq \frac{2n^2}{f_{k+1}} \left(f_k - f_{k+1} + \gamma \log n\right)} \geq 1 - n^{-\gamma}.\]
    Since $f_{k+1} \geq n \cdot 2^{-k}$ and $f_k - f_{k+1}  \leq n 2^{-k}$,
        this is further upper-bounded by
        \[\frac{2n^2}{n 2^{-k} }\bc{n 2^{-k} + \gamma \log n}
            = 2n^2 + 2\gamma \cdot 2^k \log n. \qedhere\]
\end{proof}

\section[Self-Stabilizing Ranking]{Self-Stabilizing Ranking}
\label{sec:ss_protocol}

In this section we present our self-stabilizing ranking protocol called $\StableRanking$,  which is based on $\Ranking$.  
Starting from \emph{any} possible configuration, it ranks all agents in $O(n^2\log n)$ interactions, w.h.p.
However, self-stabilization comes at a cost: it increases the memory complexity by an additive $O(\log^2 n)$ number of states.
 The core idea behind the protocol is to run the ranking protocol from the previous section and reset it whenever we detect an error, i.e., a doubly assigned label, or if the protocol does not make any progress.
To this end, $\StableRanking$ is divided into three \emph{subprotocols}, $\FastLeaderElect$, 
$\ResetProt$, and $\Ranking+$.
Both $\FastLeaderElect$ and  $\Ranking+$ are randomized and rely on a synthetic random coin (cf.\ \cite{DBLP:conf/soda/BerenbrinkKKO18}). 
In our state space, the coin is implemented under the variable $\Coin{v} \in \{0,1\}$ that is flipped every time the agent is activated.
Intuitively speaking, in each iteration, the coin shows \emph{heads} ($\Coin{v} = 1$) with probability roughly ${1}/{2}$, and \emph{tails} ($\Coin{v} = 0$) otherwise after \emph{warming up} for $O(n \log \log n)$ steps.
The protocol $\StableRanking$ uses the state space shown in \cref{alg:stable_ranking_full}.

The remainder of this section is structured as follows. We first describe the subprotocols $\ResetProt$, $\FastLeaderElect$, and $\Ranking+$ in  \cref{sec:resetprot_description,sec:simpleleader,sec:ranking_plus_description}, and then we present the analysis of our algorithm in \cref{sec:ss_ranking_analysis}.

\begin{figure*}
\let\columnwidth\textwidth

\newcommand*{\graybox}[2]{\raisebox{0.5em}{\colorbox{gray}{\parbox{#1}{\centering\color{white}#2}}}}

\def\preAlgo{
This protocol uses the following set of states, where $\uplus$ is the disjoint union of two sets:
\vspace{-1ex}
\setlength{\fboxsep}{2pt}
\[
Q =
\arraycolsep=0.5pt
\begin{array}{@{}ccccccccccccccccccc@{}}
  \graybox{3em}{\mathstrut\textsc{Rank}} & &
  \graybox{3em}{\mathstrut\textsc{Coin}} & &
  \graybox{10em}{\mathstrut\textsc{PropagateReset}} & &
  \graybox{6em}{\mathstrut\textsc{FastLE}} & &
  \graybox{10em}{\mathstrut\textsc{Ranking+}} \\
\underbrace{[n]}_{\texttt{rank}} & \uplus &
\underbrace{\{0, 1\}}_{\texttt{coin}} & \times \big( &
\underbrace{[\Theta(\log n)]}_{\texttt{resetCount}}  \times \underbrace{[\Theta(\log n)]}_{\texttt{delayCount}}
& \uplus &
\underbrace{[\Theta(\log^2 n)]}_{Q_\mathrm{SLE}} & \uplus &
\underbrace{[\Theta(\log n)]}_{\texttt{aliveCount}}  \times
\underbrace{[\Theta(\log n)]}_{\substack{\text{non-$\texttt{rank}$ state(s) } \\ \text{from \Ranking}}}
\big) \\
&&&&Q_{\mathsf{Reset}}&&Q_{\mathsf{LeaderElect}} && Q_{\mathsf{Main}}
\end{array}
\]}

\begin{algorithm}[H]{\StableRanking($u$, $v$).\label{alg:stable_ranking_full}}
execute $\PropagateReset(u, v)$ /*if applicable, propagate resets and transition into leader election*/

if $\LeaderDone u \neq \bot \neq \LeaderDone v$ then
    execute $\textsc{ElectLeader}(u, v)$ /*elect leader, handle leader's transition to waiting*/

if $\LeaderDone w \neq \bot$ and $X(x) \in \QM$ for a $\{w, x\} = \{u, v\}$ /*agent executing leader election meets agent executing main protocol*/
    set all of $w$'s state except $\Coin w$ to $\bot$ /*$w$ forgets its leader election states and becomes a phase agent*/
    $\Phase w, \LivenessCount w \gets 1, \Lmax$

if $X(u) \in \QM$ and $X(v) \in \QM$ /*when two agents having main states interaction, execute the extended ranking protocol*/
    execute $\Ranking$+($u$, $v$)

if $\Coin{v} \neq \bot$ then /*toggle $v$'s coin if it has one*/
    $\Coin{v} \gets 1 - \Coin{v}$
\end{algorithm}
\end{figure*}

\subsection{$\ResetProt$}
\label{sec:resetprot_description}

First, we will briefly discuss the resetting protocol $\ResetProt$ from \cite{DBLP:conf/podc/BurmanCCDNSX21}, which we use here in an almost black-box-like manner.
$\ResetProt$ is responsible for restarting the protocol whenever an error in either $\Ranking+$ or $\FastLeaderElect$ occurs.
To be precise, whenever an agent detects an error, $\ResetProt$ resets the agents to a configuration where all agents start to elect a leader.
Each agent $v \in V$ has two counters $\ResetCount{v} \in [0,\Rmax]$ and $\DelayCount{v} \in [0,\Dmax]$ where $\Rmax,\Dmax \in \Theta(\log n)$.
We will fix the values of $\Rmax$ and $\Dmax$ in our analysis.

The protocol works as follows.
Based on the counters $\ResetCount{v}$ and $\DelayCount{v}$, the agents are divided into three classes.
If \mbox{$\ResetCount{v} = \bot$}, we say the agent is \emph{computing} (meaning it executes $\FastLeaderElect$ or $\Ranking+$),
and if $\ResetCount{v} > 0$, we say the agent is \emph{propagating}.
Finally, if $\DelayCount{v} > 0$ and $\ResetCount{v} = 0$, we say the agent is \emph{dormant}.

If a computing agent $v$ has to restart the protocol, it sets $\ResetCount{v} = \Rmax$  and all its other variables except for $\Coin{v}$ to $\bot$.
If $\Coin{v} \neq \bot$, the value of $\Coin{v}$ is maintained; otherwise, we initialize $\Coin{v}$ to $0$.
We call $v$ the  \emph{triggered} agent, a configuration containing a triggered agent a \emph{triggered configuration} and the set of all of these configurations $\CT$.
We will write $\TriggerReset(v)$ for a routine triggering a reset for $v$ as described.

A triggered agent $v$ starts a one-way epidemic that will eventually turn all computing agents into dormant agents as follows.
If a propagating agent $v$ interacts with computing agent $w$, it decreases its $\ResetCount{v}$ by $1$ and $w$ becomes propagating by setting $(\ResetCount{w}, \DelayCount{w}) = (\ResetCount{v},\Dmax)$ and all other values (except the coin) to $\bot$.
If two propagating agents $v$ and $w$ interact, they set $\ResetCount{v}$ and $\ResetCount{w}$ to $\max\{\ResetCount{v}, \ResetCount{w}\}-1$ (unless both are $0$).
Finally, if a propagating agent $v$ interacts with dormant agent, it decreases $\ResetCount{v}$ by $1$ and if a dormant agent $v$ interacts with an arbitrary agent, it decreases $\DelayCount{v}$ by $1$.

As soon as an agent $w$ reaches $\DelayCount{w} = 0$, it forgets the state associated with $\PropagateReset$ and initializes the state of the leader election protocol, where the value of $\Coin{w}$ is maintained.



%

\subsection{$\FastLeaderElect$}\label{sec:simpleleader}

Unfortunately, we cannot use the leader election protocol by \textcite{DBLP:conf/soda/GasieniecS18} for our self-stabilizing algorithm.
In a self-stabilizing setting we need to cope with bad initializations resulting in no leader being elected.
A simple way of dealing with this is by attaching a \emph{timer} $\LECount{v}$ to each agent. 
The timer is initialized to the maximal number $\Lmax$ of interactions that are w.h.p.\ required by any agent in the protocol---for \cite{DBLP:conf/soda/GasieniecS18}, this is $\Theta(\log^2 n)$.
Agents in the leader election phase decrement $\LECount{v}$ whenever they interact.
If an agent's timer ever reaches $0$, it triggers a reset.
However, this simple trick blows up the state space by a multiplicative factor of $\Lmax$.
For the above-mentioned protocol \cite{DBLP:conf/soda/GasieniecS18}, this would result in a state space of size $O(\log^2 n \log\log n)$, which is slightly too large.
To mitigate this, we present a simple and fast protocol $\FastLeaderElect$ (similar to the lottery game of \cite{DBLP:conf/soda/AlistarhAEGR17}).
In a nutshell, the protocol works as follows.
Again, each agent $v$ stores two bits $\LeaderBit{v}$ and $\LeaderDone{v}$ that denote whether they are the leader and finished with the protocol execution. 
Recall that all unranked agents, even the \emph{dormant} agents, have a variable called $\Coin{v}$ that is \emph{flipped} on each activation.
An agent $v$ will declare itself to be the leader ($\LeaderBit{v} = 1$) if it observes $\lceil\log n\rceil$ \emph{heads} ($\Coin{w} = 1$) on its interaction partners in a row.
This requires each agent to only store $\lceil\log n\rceil$ bits, one for each coin flip.
With constant probability, there is \emph{exactly one} agent that archives this, i.e., there is exactly one leader.
This leader starts a broadcast that lets all agents start the ranking protocol just as before.
Note that, if there are two or more agents, this will trigger a reset within $O(n^2)$ interactions w.h.p., as this produces many duplicate labels. 
Furthermore, to avoid being stuck in a configuration without leaders, we use the aforementioned variable $\LECount{v}$ to count each agent's interactions.
If it reaches $0$ on one agent before it starts ranking, the agent also triggers a reset.
\iflong A detailed description of the protocol including the pseudocode can be found in \cref{sec:apx:simple_leader_elect}.
\else A detailed description of the protocol including the pseudocode can be found in the full version \cite{fullversion}.
\fi

\subsection{$\Ranking+$}
\label{sec:ranking_plus_description}

\def\preAlgo{
This protocol uses the following set of \emph{main states}, written $\QM$, where $\uplus$ is the disjoint union of two sets:
\vspace{-1ex}
\begin{align*}
\QM = \underbrace{\left\{  1, \ldots,  n\right\}}_{\texttt{rank}} \uplus \underbrace{\{0, 1\}}_{\texttt{coin}} \times \underbrace{\left\{  1, \ldots,  \Lmax\right\}}_{\texttt{aliveCount}} \times \left( \underbrace{\left\{  1, \ldots,  \cwait\cdot\log n\right\}}_{\texttt{waitCount}} \uplus \underbrace{\left\{  1, \ldots,  \lceil\log n\rceil\right\}}_{\texttt{phase}} \right).
\end{align*}
}

\begin{figure*}
\let\columnwidth\textwidth
\begin{algorithm}[H]{\Ranking+($u$, $v$).\label{alg:stable_unsafe_silent}\label{alg:ranking_plus} %
%Here $\UnsafeSilentProt_S$ behaves just like $\UnsafeSilentProt$, except that when agents become ranked agents, phase agents, or waiting agents, the values of $\LivenessCount{v}$ and $\Coin{v}$ are adjusted as described in the text.
}
$\sepline{error detection}$
if $\Rank u = \Rank v \neq \bot$ /*if $u$ and $v$ have the same rank or two waiting agents meet, trigger a reset and do nothing else*/$\label{ln:same-rank}$
        or $(\WaitCount u \neq \bot$ and $\WaitCount v \neq \bot)\label{ln:two-waiting}$ then
    execute $\TriggerReset(u)$
    return

$\sepline{liveness checking}$
if $\LivenessCount u \neq \bot \neq \LivenessCount v$ then /*if both $u$ and $v$ check liveness, adopt maximum counter minus one\label{ln:max-minus-one-s}*/
    $\LivenessCount u, \LivenessCount v \gets \max\{\LivenessCount u, \LivenessCount v\} - 1\label{ln:max-minus-one-e}$

if $\Rank u \in \{n-1, n\}$ and $\LivenessCount v \neq \bot$ then /*when meeting an agent ranked $\geq n-1$, decrement counter if present\label{ln:n-1-or-n-s}*/
    $\LivenessCount v \gets \LivenessCount v - 1\label{ln:n-1-or-n-e}$

if $\LivenessCount v = 0$ then /*if the counter hits zero, trigger a reset and do nothing else\label{ln:alive-count-s}*/
    execute $\TriggerReset(u)$
    return$\label{ln:alive-count-e}$

if $\Coin v = 0$ then /*if $v$'s coin is $0$, reset the counter if we \emph{could} have made progress\label{ln:coin-0-s}*/
    if $\WaitCount u \neq \bot$ or $\Big(\Rank u \neq \bot \neq \Phase v$ and $\Rank u \leq \floor{n \cdot 2^{-\Phase v}}\Big)$ then$\label{ln:productive}$
        $\LivenessCount v \gets \ceil{\cliveness \cdot \log n}\label{ln:coin-0-e}$

$\sepline{base protocol}$
else if $\Coin v = 1$ then /*if $v$'s coin is $1$, execute the base protocol\label{ln:coin-1-s}*/
    execute $\UnsafeSilentProt(u, v)\label{ln:coin-1-e}$
    if $u$ became waiting in the $\UnsafeSilentProt$ protocol then
        $(\Coin u,\LivenessCount u)  \gets (0, \Lmax)$
\end{algorithm}
\end{figure*}

Similarly to $\Ranking$, each agent can have exactly one value of $\WaitCount{v}$, $\Phase{v}$, or $\Rank{v}$ \emph{not} equal to $\bot$. We again call agents having one of these values waiting agents, phase agents, and ranked agents, respectively.
A leader $\ell$ alternates between the states waiting (\texttt{waitCount} $> 0$) and being unaware leader (there is an agent $v$ with $\Rank \ell \leq n \cdot 2^{-\Phase v}$).
The counter \texttt{aliveCount} will be used to check if the protocol is still making progress.
The protocol is defined in \cref{alg:ranking_plus}.

At its core, $\Ranking+$ extends $\Ranking$ by triggering a reset (via $\PropagateReset$) as soon as one of the three following errors occurs:
First, two agents have the same rank, which is detected when they interact directly (line \ref{ln:same-rank}).
Second, more than two agents are waiting, which is also detected via direct interaction (line \ref{ln:two-waiting}).
Third, the protocol cannot assign more ranks.
As we will show later, in this case $\LivenessCount{v}$ will reach $0$ for at least one agent (line \ref{ln:alive-count-s}).

Note that it is unfortunately not always possible to detect if there are two or more leaders.
However, unless they ``accidentally'' produce a correct ranking, this case will ultimately result in one of the three kinds of errors above being produced and detected.

It remains to explain how $\Ranking+$ detects the third error. The idea of $\LivenessCount{v}$ is that it will reach the value $0$ whenever the protocol does not make any progress. The variable $\texttt{aliveCount}$ is decremented for two reasons.
First, whenever two unranked agents meet, they update their $\texttt{aliveCount}$ to the maximum of their respective counts minus $1$ (lines \ref{ln:max-minus-one-s}--\ref{ln:max-minus-one-e}).
Second, an unranked agent $u$ also decrements $\LivenessCount{u}$ when encountering an agent ranked $n-1$ or $n$ (lines \ref{ln:n-1-or-n-s}--\ref{ln:n-1-or-n-e}). This is necessary to reduce $\texttt{aliveCount}$ in the case where $u$ is the only unranked agent.
If there is no interaction that resets $\LivenessCount{v}$, the counter will eventually reach $0$.

To explain how $\Ranking+$ actually detects the third error we distinguish between two cases: the leader is either waiting or an unaware leader.
In the first case, whenever the waiting leader $\ell$ encounters a phase agent $v$, it resets $\LivenessCount{v}$.
In the second case, ideally we would like to reset the counter to its maximum value whenever the unaware leader $\ell$ assigns a new rank to an agent $v$. Unfortunately, neither $\ell$ nor $v$ have state space left to store $\texttt{aliveCount}$.
To circumvent this problem we do the following.
Whenever a ranked agent $u$ interacts with an unranked agent $v$, agent $u$  will determine if it is the unaware leader
$\ell$ (by checking if $\Rank u \leq n \cdot 2^{-\Phase v}$).
If the inequality is fulfilled $u$ decides between either assigning a rank to $v$ (if $\Coin v=1$, lines \ref{ln:coin-1-s}--\ref{ln:coin-1-e}) or setting $\LivenessCount{v}$ to $\Lmax$ (if $\Coin v = 0$, lines \ref{ln:coin-0-s}--\ref{ln:coin-0-e}).



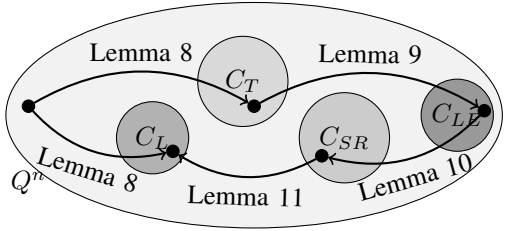
\begin{figure}[t]
    \centering

\begin{tikzpicture}[scale=0.6]

  \colorlet{colN}{gray!60}
  \colorlet{colZ}{gray!40}
  \colorlet{colQ}{gray!80}
  \colorlet{colR}{gray!10}
  \colorlet{colT}{gray!30}

  \coordinate (oN) at (-2.25, -0.5);
  \coordinate (oZ) at (2, -0.5);
  \coordinate (oQ) at (4.5, 0.0);
  \coordinate (oR) at (0, 0);
  \coordinate (oT) at (-0.25, 0.75);

\filldraw[fill=colR] (0, 0) ellipse (5.5 and 2.5);
\filldraw[fill=colQ] (oQ) circle (0.8);
\filldraw[fill=colZ] (oZ) circle (1.0);
\filldraw[fill=colN] (oN) circle (0.8);
\filldraw[fill=colT] (oT) circle (1.0);

\node (N) at (oN) {$C_L$};
\node (Z) at (oZ) {$C_{SR}$};
\node (Q) at (oQ) {$C_{LE}$};
\node (R) at (oT) {$C_T$};
\node at (-5,-1.5) {$Q^n$};

\node[draw, circle, fill, scale=0.5] (A) at (-5, 0.2) {};
\node[draw, circle, fill, scale=0.5] (B) at (0, 0.2) {};
\node[draw, circle, fill, scale=0.5] (C) at (5.1, 0.1) {};
\node[draw, circle, fill, scale=0.5] (D) at (1.5, -0.9) {};
\node[draw, circle, fill, scale=0.5] (E) at (-1.8, -0.8) {};

\path[->, thick] (A) edge[bend left] node[midway,sloped,above]{{\cref{lemma:convergence_bad_case}}} (B);
\path[->,  thick] (A) edge[bend right] node[midway,sloped,below]{{\cref{lemma:convergence_bad_case}}} (E);
\path[->,  thick] (B) edge[bend left] node[midway,sloped,above]{{\cref{lem:recovery}}} (C);
\path[->,  thick] (C) edge[bend left] node[midway,sloped,below]{{\cref{lemma:convergence_leader_elect}}} (D);
\path[->,  thick] (D) edge[bend left] node[midway,sloped,below]{{\cref{lemma:convergence_good_case}}} (E);
\end{tikzpicture}
\vspace{-1ex}
\caption{A high-level overview of the self-stabilizing algorithm.}
\label{fig:enter-label}
\end{figure}
\subsection{Analysis}
\label{sec:ss_ranking_analysis}

In this section we show \cref{thm:ss-ranking}. Similarly to the proof of \cref{thm:non-self-stab-ranking} we define a set of configurations which are safe entry states of the subprotocols of our protocol.

The set $\CLE$ contains the configurations from which we can safely start $\FastLeaderElect$. We call this set leader election configurations. 
The set $\CT$ contains the triggered configurations.
These are configurations in which we start $\PropagateReset$. 
The set $\CSRPlus$ contains safe ranking configurations defined as the configurations from which we can safely call $\Ranking+$. 
Finally, $\CL$ contains all legal configurations in which all agents have a unique rank.
The following proof of \cref{thm:ss-ranking} essentially tracks
movement of the protocol through these sets until a configuration in $\CL$ is reached, see \cref{fig:enter-label} for an overview of these sets.


\begin{proof}[Proof of \cref{thm:ss-ranking}]
For the state space, see that the three sub-protocols have $O(\log^2 n)$ overhead states.
Together with the coin and the ranks, we have $|Q| = n + O(\log^2 n)$ as required.

The analysis of the correctness and running time of the protocol is split into \cref{lemma:convergence_bad_case,lem:recovery,lemma:convergence_leader_elect,lemma:convergence_good_case}.
In \cref{lemma:convergence_bad_case}, we show that when $\vX_t \not\in C_L$ is an arbitrary configuration, within $O(n^2 \log n)$ interactions the protocol either reaches configuration in
$\CT$ or $\CL$, w.h.p.
\Cref{lem:recovery} shows that, when the protocol is an arbitrary configuration of $\CT$, it reaches a configuration from $\CLE$ within $O(n \log n)$ interactions, w.h.p.
\Cref{lemma:convergence_leader_elect} shows that, when the protocol is in an arbitrary configuration of $\CLE$, it reaches a configuration from $\CSRPlus$ within $O(n^2 \log n)$ interactions, w.h.p.
Finally, \cref{lemma:convergence_good_case} shows that from a configuration in $\CSRPlus$, the protocol will reach a correct ranking configuration $\CL$ in $O(n^2 \log n)$ interactions, w.h.p.,
    directly implying 
the correctness and running time.

We conclude with the protocol's closure.
When the protocol is in a legal configuration $\vX_t \in C_L$, all pairs of agents $u$ and $v$ have distinct ranks. In this case, they do not change their states. This can be seen by inspecting the protocol as two ranked agents only perform an action when they have equal rank.
Therefore, $\vX_{t+1} = \vX_{t}$ and the protocol fulfills the closure property and is silent.
\end{proof}

\begin{lemma}
\label{lemma:convergence_bad_case}
Let $c_1$ be a sufficiently large constant, and assume that $\vX_t$ is an arbitrary configuration not in $\CL$.
Then, w.h.p., there is a $\tau \leq c_1 \cdot n^2 \log n$ such that either $\vX_{t+\tau} \in \CL$ or one the configurations $\vX_{t}, \ldots, \vX_{t+\tau}$ contains a triggered agent (i.e., is in $\CT$).
\end{lemma}
\begin{proof}[Proof sketch]
First, we show that after a \emph{preparation phase} having $O(n \log^2 n)$ rounds,
    the protocol will be in a configuration where all agents are in a main state (i.e., either waiting agents, phase agents, or ranked agents).
The remainder of the proof is then divided into two parts, both using a potential function.
This potential function allows us to treat the various ways to make progress (assigning ranks, advancing saved phases, or resetting) in a unified manner without excessive case distinctions.

We call a pair of agents $u \neq v$ a \emph{productive pair} if it fulfills the condition in line \ref{ln:productive} of \cref{alg:ranking_plus}.
That is, the protocol \emph{could} make progress if the phase agent's coin shows $1$.
    Our potential function $\Phi_t$ is then defined as
    $0$ if there is no productive pair or there is a resetting agent in $\vX_t$ and as $\sum_{v \in [n]\colon \PhaseT t v \neq \bot} 2^{-\PhaseT t v}$ otherwise.

    In the first part of the proof we show that the potential drops to zero within $O(n^2 \log n)$ rounds w.h.p. 
    In the second part of the proof we then show that once the potential has reached zero, the protocol is either in a configuration in $\CL$ or will reset within $O(n^2 \log n)$ further interactions w.h.p. 
    This directly implies \cref{lemma:convergence_bad_case}.

\medskip

For the first part of the proof we define a notion of \emph{good time steps}, for which we can show an expected drop of the potential $\Phi$ in $\Omega(\Phi / n^2)$, leading to a geometric decay with decay factor $1 - \Omega(n^{-2})$.
    There are multiple ways in which a time step can be good. 
A time step is good if there is a directly detectable error, i.e., there are two agents having the same rank, or two waiting agents.
In this case there is at least an $n^{-2}$ probability of detection, in which case the potential immediately drops to zero.
A time step is also good if an agent can increase its saved phase by some interaction.
If this is the case, it must be true in particular for an agent $v$ with the lowest saved phase.
Letting $s$ be the number of phase agents in the configuration, $v$ contributes at least $\Phi / s$ to the potential, and this contribution will at least halve when $v$ increases its phase.
By case analysis, one can see that the probability of this occurring is in $\Omega(s/n^2)$.
$s$ cancels out between the probability and the drop in potential in this event, and we get the desired drop in potential in expectation.
Lastly, a time step is good if an agent can be ranked in some interaction, and a sufficient proportion of phase agents have their coin showing $1$ (so that they can actually get ranked).
Assuming we are not in the previous case as well, \emph{all} phase agents have the same phase and may be ranked in an interaction.
So an agent gets ranked with probability in $\Omega(s / n^2)$.
As each phase agent contributes exactly $\Phi/s$ to the potential, and this potential contribution drops to $0$ when it ceases to be a phase agent, we get the desired drop in potential as in the previous case.

Finally, we need to show that there are enough good time steps within a time interval of size $O(n^2 \log n)$.
This involves the analysis of the synthetic coin among phase agents.
Because this subpopulation shrinks (due to agents getting ranked) and can become quite small (even $o(\log n)$),
    and agents are removed from the subpopulation depending on the current value of the coin,
    we cannot use established techniques for this analysis.
\iflong The full proof can be found in \cref{sec:proof_bad_case}. \fi
\medskip

For the second part of the proof, we need to show that if the potential~$\Phi$ has reached zero but the protocol is not in a configuration in $\CL$, a reset is triggered within $O(n^2 \log n)$ interactions.
The challenge is to show that the protocol triggers a reset if there are no productive pairs.
Here, we proceed by case distinction: either there are two agents with the same rank, a single unranked agent, or multiple unranked agents.
In the proof we use an argument adapted from \cite[Lemma 3.3]{DBLP:conf/podc/BurmanCCDNSX21}, which, in turn, is adapted from \cite[Lemma 1]{DBLP:conf/dna/AlistarhDKSU17}.
\end{proof}

\begin{lemma}\label{lem:recovery}
   Let $c_2$ be a sufficiently large constant, and assume that $\vX_t$ is an arbitrary configuration in $\CT$, i.e., containing a triggered agent.
   Then, w.h.p., there is a $\tau \leq c_2 \cdot n \log n$ such that  $\vX_{t+\tau} \in \CLE$.
\end{lemma}

    \textit{Proof sketch.}The lemma follows more or less directly from the correctness of the $\PropagateReset$ protocol \cite{DBLP:conf/podc/BurmanCCDNSX21},
        we also need to ensure that the synthetic coin used by the leader election is sufficiently ``warmed up'' by the time the reset has run its course. 
        However, this can be achieved by letting the agents be dormant long enough.
\iflong The detailed proof is given in \cref{sec:proof_recovery}. \fi

\begin{lemma}
\label{lemma:convergence_leader_elect}
Let $c_3$ be a sufficiently large constant, and assume that $\vX_t$ is an arbitrary configuration in $\CLE$.
Then, w.h.p., there is a $\tau \leq c_3 \cdot n^2 \log n$ such that $\vX_{t+\tau} \in \CSRPlus$.
\end{lemma}
\noindent \textit{Proof sketch.} Here we have to show that, from an arbitrary leader election configuration, we reach a safe ranking configuration. Note that our simple leader election protocol has a constant failure probability (it can elect none or several leader). 
If the leader election protocol does not elect a leader is starts over again and we have to how that w.h.p.\ this does not happen too often.  
If several leaders are selected, all these leaders start the first phase of the ranking algorithm. 
Thus, there will be agents receiving the same rank. This will be detected within $O(n^2)$ interactions by $\Ranking+$ resulting in a reset.
\iflong Details of the proof can be found in \cref{sec:proof_leader_elect}. \fi
\begin{lemma}
\label{lemma:convergence_good_case}
Let $c_4$ be a sufficiently large constant, and assume that $\vX_t$ is an arbitrary configuration in $\CSRPlus$.
Then, w.h.p., there is a $\tau \leq c_3 \cdot n^2 \log n$ such that $\vX_{t+\tau} \in \CL$.
\end{lemma}
\noindent \textit{Proof sketch.} 
This proof follows along the lines of the analysis of the non-self-stabilizing protocol in \cref{sec:unstable_ranking_analysis}.
The main added difficulty is to show that w.h.p.\ no agent $v$ reaches a state where $\LivenessCount{v} = 0$ which would inadvertently trigger a reset.
Similarly to the proof of \cref{lemma:convergence_bad_case},
this requires an analysis of the synthetic coin among the shrinking and potentially tiny subpopulation of phase agents,
and additionally the analysis of one-way epidemics in the same setting.
\iflong The detailed proof is given in \cref{sec:proof_good_case}. \fi
\section{Simulation Results}
\begin{figure}
    \input{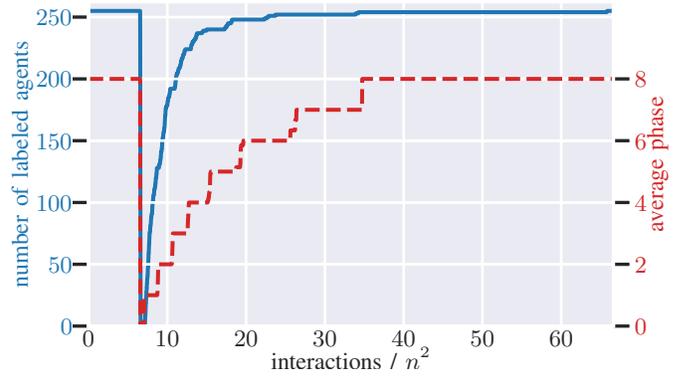}
    \caption{Number of ranked agents (blue), and average of the phase counters stored by unranked agents (red, dashed), as a function of the number of interactions. The protocol (for $n=256$) is initialized as follows: 255 agents are ranked (with ranks $2, \ldots, 256$), and one agent is a phase agent with maximum liveness counter.}
    \label{fig:run256}
\end{figure}

\begin{figure}
    \input{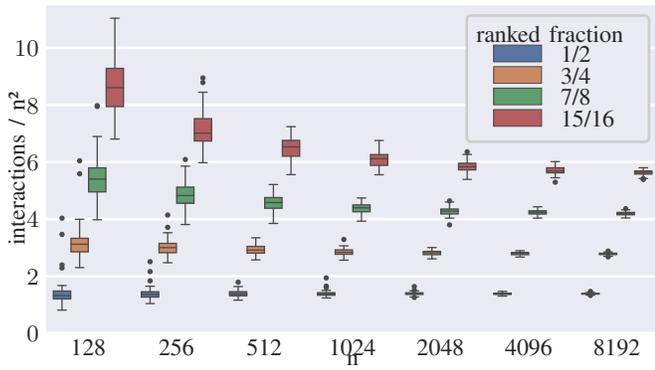}
    \caption{Number of interactions (normalized by $n^2$) needed to reach a state in which $c \cdot n$ agents are ranked, starting from the following configuration: there is one agent in rank 1 (the unaware leader), and all other agents are still in a leader election state.
    We performed 100 simulations per value of $n \in \{2^i \mid i \in \N, 7 \leq i \leq 13\}$.}
    \label{fig:ranktimes}
\end{figure}

{
We implemented a simulation of our population protocol in Rust (with $\cwait=2$ and $\cliveness=\Dmax / \log_2(n) =4$).
In \cref{fig:run256} we show how the protocol resets and then quickly resumes assigning ranks starting from an invalid initialization.
The chosen initialization can be considered to be worst-case as it needs $\Theta(n^2 \log n)$ interactions to reset (in expectation).
The figure shows that most of the runtime is taken up by ranking the final few agents, with successive phases taking increasingly longer.
This is to be expected, as the process is a coupon collection process.
Accordingly, 
it should take about as long to rank half the agents as it takes to rank the next quarter, the next eight, and so on.
\cref{fig:ranktimes} confirms this; there, we consider the number of interactions to rank constant fractions of agents for various~$n$ and fractions.
After 
$\Theta(n^2)$ interactions, constant fractions of agents are ranked, much faster than the $\Theta(n^2 \log(n))$ interactions needed to rank \emph{all} agents.

}

\section{Conclusion}

We present a self-stabilizing protocol for the population protocol model that solves the ranking problem.
Our protocol is silent and requires $O(n^2 \log n)$ interactions w.h.p.\ using $n+O(\log^2 n)$ states.
It is an open question to solve the ranking problem either in $\Theta(n^2)$ interactions in expectation using $n + O(\log(n))$ states or in $\Theta(n^2 \log n)$ interactions w.h.p.\ using $n + o(\log n)$ states.
Another question is if it is possible to improve on the $\Theta(\log^2 n)$ overhead for self-stabilizing leader election while keeping the number of interactions at $O(n^2 \log n)$.
Finally, it is also open whether time-optimal (non-silent) protocols exist that use only subexponentially many states, improving upon the 
protocol by \textcite{DBLP:conf/podc/BurmanCCDNSX21}.

\section*{References}
\printbibliography[heading=none]

\iflong

\appendices
\newpage
\onecolumn

\section{Tail Bounds}

\label{apx:tail-bounds}

\noindent First, we give a tail bound for the negative binomial distribution.

\begin{lemma}\label{lem:negbin_tailbound}
    Let $X \sim \NegBin(r, p)$ have negative binomial distribution with parameters $r \geq 1$ and $p \in (0, 1)$.
    \begin{enumerate}
        \item For ${\gamma > 0}$ and ${n \geq 1}$, \(\Prob{X > \frac{1}{p} \cdot 2 \cdot \bc{r + \gamma\log n}} \leq n^{-\gamma}.\)
        \item \(\Prob{X \leq \frac{1}{2} \cdot \frac{r}{p}} \leq \exp\bc{-\,\frac{r}{6}}.\)
    \end{enumerate}
\end{lemma}



Next, we give a bound on the running time of coupon collection.
We let $\CouponCollector(n)$ be the distribution of the number of trials needed in the coupon collector's problem with $n$ coupons. 
The following lemma is a standard upper tail bound on this distribution.
\begin{lemma}[cf., e.g., Section 3.3.1 in \cite{mitzenmacher2017probability}]\label{lem:coupon_collector_tailbound}
    Let $1 \leq k \leq n$, and $\gamma > 0$.
    Then for $X \sim \CouponCollector(k)$ we have
        \[\Prob{X > k\bc{\log(k) + \gamma \log n}} \leq n^{-\gamma}.\]
\end{lemma}

Finally we give upper and lower bounds for the running time of one-way epidemics based on Lemma 2 from \cite{DBLP:journals/dc/AngluinAE08a}.
%
%
Write $\OWE(n, m)$ for the distribution of the number of interactions needed
    to perform a one-way epidemic among a subset of $m$ agents in a total population of $n$ agents,
    where one of the $m$ agents is initially infected.

\begin{lemma}\label{lem:subset_broadcast_upper_tail}
    For $X \sim \OWE(n,m)$ for $2 \leq m \leq n$, and all $\gamma > 0$, we have
        \[\Prob{X > 3 \frac{n^2}{m} \cdot \bc{\log(m) + 2\gamma \log\bc{n}}} \leq 2n^{-\gamma}.\]
\end{lemma}

\section{Omitted Proofs of \cref{sec:unstable_ranking}}
\label{apx:sec:unstable_ranking}

In this appendix we present the omitted proofs from \cref{sec:unstable_ranking}. We first state a lemma that establishes the correctness of leader election based on the protocol by \textcite{DBLP:conf/soda/GasieniecS18}.

\begin{lemma}[cf.\ Lemma 2.4 in \cite{DBLP:conf/podc/BerenbrinkKR19}, citing \cite{DBLP:conf/soda/GasieniecS18}]\label{lem:leszek_clock}
    There is a population protocol using $O(\log\log n)$ states w.h.p.\ electing a unique leader in $O(n \log^2 n)$ interactions w.h.p.: after at most $O(n \log^2 n)$ interactions, there is, w.h.p., an agent $\ell$ with $\LeaderDone{\ell} = 1$ and $\LeaderBit{\ell} = 1$, and at that time, w.h.p., all other agents $v \neq \ell$ have $\LeaderBit{u} = 0$.
\end{lemma}

We now give the full technical proofs for the lemmas from \cref{sec:unstable_ranking}.

\begin{proof}[Proof of \cref{leader}]
    Recall that at time $t=0$,
        all agents $v$ start with $\qLE v = q_0$ where $q_0$ is the initial state of the leader election protocol as described by \cite{DBLP:conf/podc/BerenbrinkKR19},
        and with $\LeaderDone v = 0$.
    Whenever two agents with $\qLE v \neq \bot$ interact,
        they follow the transition function of the leader election (see lines \ref{ln:le-s}--\ref{ln:le-e} of \cref{alg:unsafe_silent_plus}).
    As long as no agent $\ell$ transitions into a state with $\LeaderBit \ell = 1$ and $\LeaderDone \ell = 1$,
        these are the only transitions taking place,
        since the one-way epidemic setting $\qLE v$ and $\LeaderDone v$ to $\bot$ (see lines \ref{ln:oe-s}--\ref{ln:oe-e}) requires an initially infected agent, which is $\ell$ (see lines \ref{ln:le-s-s}--\ref{ln:le-s-e}).
    By \cref{lem:leszek_clock},
        w.h.p.\ this takes place at a time $\tau$ with $\tau = O(n \log^2 n)$;
        the leader $\ell$ is then immediately transitioned into the state where $\WaitCount \ell = \ceil{\cwait \cdot \log n}$ in the same interaction (line \ref{ln:le-s-e}).
    And at that point, by \cref{lem:leszek_clock}, all other agents $v \neq \ell$ have $\LeaderBit v = 0$.
    So after the interaction at time $\tau$, the configuration fulfills all the requirements of a safe ranking configuration, yielding the claim.
\end{proof}

\begin{proof}[Proof of \cref{lemmacorrectness_non_ss}]
We show for any constant $\gamma > 0$ that, assuming that $\vX_t$ is an arbitrary configuration in $\CSR$,
    there is, with probability at least $1 - O(\log n \cdot n^{-\gamma})$,
    a $\tau \leq 4(\cwait + 5\gamma + 1)n^2 \ceil{\log_2 n}$
    such that $\vX_{t+\tau} \in \CL$,
    which implies the claim.

First, we show by induction that for all $1 \leq k \leq \ceil{\log_2 n} \eqqcolon k_\mathrm{max}$,
    with probability at least $p_k = 1 - 6 (k-1) n^{-\gamma}$ there is a $\tau \leq \tau_{k,\mathrm{max}} = 4(k-1)n^2 + (\cwait + 5 \gamma) n \log n \cdot \sum_{k'=1}^{k-1} 2^k$
        such that $\vX_{t+\tau} \in \CWFWait{k}$.
For $k=1$, the statement is true with $\tau_k = 0$ since $\vX_t \in \CSR = \CWFWait{1}$ by assumption and definition of $\CWFWait{1}$.

So assume the statement is true for some $k < k_\mathrm{max}$.
Then with probability at least $p_k$, there is a $\tau_k \leq \tau_{k,\mathrm{max}}$ such that $\vX_{t+\tau} \in \CWFWait{k}$.
Assuming this is the case, by \cref{lem:unsafe_waiting_subphase},
    with probability at least $1 - 5n^{-\gamma}$
    there is a $\tau' \leq (\cwait + \gamma) 2^k n \log n$ such that $\vX_{t+\tau+\tau'} \in \CWFRank{k}$.
And assuming that \emph{this} is the case,
    there is, by \cref{lem:unsafe_ranking_subphase},
        with probability at least $1 - n^{-\gamma}$,
        a $\tau'' \leq 2n^2 + 2\gamma 2^k n \log n$
        such that $\CWFWait{k+1}$ (since $k < k_\mathrm{max}$).
So overall by the union bound,
    with probability at least $p_k - 6n^{-\gamma} = p_{k+1}$,
        there is a $\tau''' = \tau + \tau' + \tau''
            \leq \tau_{k,\mathrm{max}} + 4n^2 + (\cwait + 5\gamma) 2^k n \log n
            = \tau_{k+1,\mathrm{max}}$
        such that $\vX_{t+\tau'''} \in \CWFWait{k+1}$.

Using the statement proven by induction
    and analogously applying \cref{lem:unsafe_ranking_subphase,lem:unsafe_waiting_subphase} a final time for $k = k_\mathrm{max}$,
    we see by the union bound
    that with probability at least $1 - 6 k_\mathrm{max} n^{-\gamma}$,
    there is a $\tau$ with $\vX_{t+\tau} \in \CL$,
    where
        \begin{align*}
        \tau
           &\leq \tau_{k_\mathrm{max},\mathrm{max}} + 4n^2 +  (\cwait + 5\gamma) 2^{k_\mathrm{max}} n \log n
            = 4 k_\mathrm{max} n^2 + (\cwait + 5\gamma)n \log n \cdot \sum_{k=1}^{k_\mathrm{max}} 2^k.
        \end{align*}
Now as $2^{k_\mathrm{max}} = 2^{\ceil{\log_2 n}} < 2n$,
    and the geometric series sums to $\sum_{k=1}^{k_\mathrm{max}} 2^k = 2^{k_\mathrm{max} + 1} - 1 < 2 \cdot 2^{k_\mathrm{max}} < 4n$,
    we have
    \[\tau \leq 4n^2 \ceil{\log_2 n} + 4 (\cwait + 5\gamma) n^2 \log n
        \leq 4(\cwait + 5\gamma + 1)n^2 \ceil{\log_2 n},\]
    as claimed.
\end{proof}

\section{Description of $\FastLeaderElect$}
\label{sec:apx:simple_leader_elect}

In this section, we describe $\FastLeaderElect$, a simple protocol that elects a leader with constant probability and otherwise triggers reset.
The protocol and its state space are formally defined in \cref{{alg:leader_elect}}.
Each agent $v \in V$ has a counter $\LECount{v} \in [0,\Lmax]$, a counter $\CoinCount{v} \in [\lceil\log n\rceil]$, and two flags $\LeaderDone{v} \in \{0,1\}$ and $\LeaderBit{v} \in \{0,1\}$ as variables.
Slightly abusing notation, $\QLE$ is again the state space used by the leader election.
Just as $\Rmax$ and $\Dmax$, we will bound $\Lmax = \cliveness \cdot \log n$ for some $\cliveness > 0$ in the analysis.
As the bounds for the other variables are fixed, it holds $|\QLE| \in O(\log^2 n)$.

Recall that all unranked agents, even the \emph{dormant} agents, have a variable called $\Coin{v}$ that is flipped on each activation.
In a nutshell, the protocol works as follows: an agent $v \in V$ will declare itself to be the leader ($\LeaderBit{v} = 1$) if it observes $\lceil\log n\rceil$ \emph{heads} ($\Coin{w} = 1$) in a row.
Furthermore, to avoid being stuck in a configuration without leaders, we use $\LECount{v}$ to count each agent interactions.
If it reaches $0$, the protocol triggers a reset.

More precisely, the protocol uses the following interactions.
Whenever an agent $v \in V$ interacts with another agent $w \in V$, the following happens.
If $\LeaderDone{v} = 1$, $v$ has already decided if it is a leader and will not consider the coin.
Otherwise, it observes its coin $\Coin{w}$.
If $\Coin{w} = 0$, it sets $\LeaderDone{v} = 1$.
If $\Coin{w} = 1$, it decrements $\CoinCount{v}$ by $1$.
If $\CoinCount{v}$ reaches $0$, $v$ has seen $\lceil\log n\rceil$ \emph{heads} in a row, and becomes leader.
To this end, it sets $\LeaderBit{v} = 1$ and $\LeaderDone{v} = 1$.
Finally, $v$ decrements $\LECount{v}$ by $1$.
If it reaches $\Lmax/2$ and $v$ is the leader, $v$ assumes that it is the unique leader.
Therefore, it transitions to the main protocol by turning into a waiting agent.
This will start a one-way epidemic that lets all other agents enter a state from $\QM$.
Finally, if $\LECount{v}$ reaches $0$ and $v$ is not a leader, $v$ assumes that no leader was elected.
Therefore, it triggers a reset.

We assume that every agent $v$ starts the protocol in a state $q_{0,i}$ for $i \in \{0, 1\}$ where
\[ (\LECount{v},\CoinCount{v},\LeaderDone{v},\LeaderBit{v},\Coin{v})\!=\!(\Lmax,\lceil\log n\rceil,0,0,\bot),\]
and all other fields are $\bot$.





\def\preAlgo{The protocol uses the following state space.\begin{align*}
    \QLE = \underbrace{\left\{  1, \ldots,  \Lmax \right\}}_{\texttt{LECount}} \times \underbrace{\left\{  1, \ldots,  \lceil\log n\rceil\right\})}_{\texttt{coinCount}} \times \underbrace{\{0,1\}}_{\texttt{leaderDone}} \times \underbrace{\{0,1\}}_{\texttt{isLeader}}
\end{align*}}
\begin{figure*}
\columnwidth=\textwidth
\begin{algorithm}[H]{\FastLeaderElect($u$, $v$):\label{alg:leader_elect}}
$\sepline{leader election phase}$
$\LECount{u} \gets \LECount{u} - 1$

if $\FlipBit{v} = 0$ then $\LeaderDone{u} \gets 1$ /*if random coin is $0$, $u$ will not be leader*/

if $\LeaderDone{u} = 1$ then return /*if $\LeaderDone{u}$ is $1$, do nothing*/

if $\CoinCount{u} > 0$ then
    $\CoinCount{u} \gets  \CoinCount{u} - 1$ /*$u$ counts coins with value $1$*/
else /*$u$ observed $\lceil \log n \rceil$ coins with value $1$*/
    $\LeaderBit{u} \gets 1$  /*$u$ becomes leader*/
    $\LeaderDone{u} \gets 1$  /*$u$ stops looking further*/

$\sepline{transition to main phase}$
if $\LECount{u} \geq {\Lmax}/{2}$ and $\LeaderBit{u}=1$ then /*leader was elected fast enough*/
    $\LECount{u}, \LeaderDone{u}, \LeaderBit{u}, \CoinCount{u} \gets \bot$
    $(\WaitCount{u}, \LivenessCount{u}) \gets (\ceil{\cwait \log n}, \Lmax)$ /*$u$ starts main phase as waiting leader*/
    return

if $\LECount{u} = 0$ then  /*leader was not elected fast enough*/
    $\LECount{u}, \LeaderDone{u}, \LeaderBit{u}, \CoinCount{u} \gets \bot$
    execute $\TriggerReset(u)$ /*trigger a reset*/
\end{algorithm}
\end{figure*}

\section{Omitted Proofs from \cref{sec:ss_ranking_analysis}}

\subsection{Proof of \cref{lemma:convergence_bad_case}}
\label{sec:proof_bad_case}

In this section, we show that starting from an \emph{arbitrary} configuration, the protocol will either reach a correct ranking ($\in \CL$) or trigger a reset within $O(n^2 \log n)$ interactions.
To this end, we divide the execution of the protocol into two phases, the \emph{preparation phase} and the \emph{main phase}.
With $\QM$ being the main states defined at the beginning of \cref{sec:ranking_plus_description},
    we let $\CM$ be the set of configurations where all agents have a state in $\QM$,
    calling them \emph{main configurations}.
We say that the protocol is in the main phase when it is in a main configuration.
We let $\CP$ be the complement of $\CM$, i.e., the set of configurations where at least one agent has state not in $\QM$, calling these \emph{prep configurations}.
We say that the protocol is in the prep phase when it is in a main configuration.
As $\CM$ and $\CP$ partition all configurations by definition, so the system is always in one of the two phases.
The following observation will be useful throughout.
\begin{observation}\label{obs:cmain_stays_unless_reset}
    Assuming that $\vX_t$ is an arbitrary configuration in $\CM$,
        $\vX_{t+1}$ will either contain a triggered agent
            or be in $\CM$.
\end{observation}

First, we will show that we quickly, i.e., within only $O(n \log^2 n)$ interactions, leave the preparation phase and start with the main phase.
\begin{lemma}\label{lem:prep_to_main}
Let $c$ be a sufficiently large constant, and assume $\vX_t \in \CP$.
Then with probability $1 - O\left(\frac{1}{n}\right)$, there is a $\tau \in [c \cdot n \log^2 n]$ such that either $\vX_{t + \tau} \in \CM$ or that the protocol resets at time $t + \tau$.
\end{lemma}

\begin{proof}
If there are still propagating agents, by the properties of the reset protocol, these agents will eventually become dormant.
Thus, within $O(n \Rmax)$ steps, we will reach a configuration where all agents are dormant, or in a state of $\QLE$ of $\FastLeaderElect$, or in a state from $\QM$.
Note that any interaction with a dormant or electing agent with an agent from $\QM$ will change both agent's states to a state from $\QM$.
Thus, if one agent is in $\QM$, all agents will be in $\QM$ after $O(n\log n)$ steps.
Therefore, it remains to show that either one agent enters $\QM$ or triggers a reset.
Note that all dormant agents decrease their $\DelayCount{v}$ by one on every interaction.
Thus, within $O(n \Dmax)$ steps, we will reach a configuration where all agents are in a state of $\QLE$ of $\FastLeaderElect$.
All leader-electing agents decrease their $\LECount{v}$ by one on every interaction.
Thus, within $O(n \Lmax)$ steps, they will either trigger a reset or one agent becomes leader and switches to $\QM$ (which triggers an epidemic that turns all agents to $\QM$).
\end{proof}

Recall from the proof sketch of \cref{lemma:convergence_bad_case} that we call a pair of agents $u \neq v$ a \emph{productive pair} if it fulfills the condition in line \ref{ln:productive} of \cref{alg:ranking_plus} (ignoring the coin),
    i.e., if it is a pair where the protocol \emph{could} make progress if the phase agent's coin shows $1$.
There are hence two ways for a pair of agents to be productive.
    Either $u$ may assign a rank to $v$ when interacting (ignoring the valid of $\Coin v$ here), i.e., when $\RankT t u \neq \bot \neq \PhaseT t v$, and $\RankT t u \leq \floor{n \cdot 2^{-\PhaseT t v}}$, in which case we also call it a \emph{rank-assigning pair};
    or $u$ is waiting and $v$ has a phase,
        i.e., $\WaitCountT t u \neq \bot \neq \PhaseT t v$.
Recall that we define the potential $\Phi_t$ as $0$ if there is no productive pair or there is a resetting agent in $\vX_t$ and as \[\Phi_t = \sum_{v \in [n]\colon \PhaseT t v \neq \bot} 2^{-\PhaseT t v}\] otherwise.

Our main lemma is two-fold.
    The first part states that the potential will drop to $0$ within $O(n^2 \log n)$ interactions w.h.p.\ when it is currently in a main configuration. We prove this part in \cref{sec:proof:lem:all_main_state_config_dies_quickly}.
    The second part states that once the potential has hit $0$, the protocol will either reach a stable configuration or reset within a further $O(n^2 \log n)$ interactions w.h.p.
    We prove that part in \cref{sec:proof:lem:dead_configs_end_quickly}.

\begin{lemma}\label{lem:main_states_reset_quickly}
    Let $c$ be a sufficiently large constant independent of $\cwait$ or $\cliveness$.
    \begin{enumerate}
        \item Assume $\vX_t$ is an arbitrary configuration in $\CM$. Then w.h.p.\ there is a $\tau \leq c \cdot \cwait \cdot n^2 \log n$ such that $\Phi_{t+\tau} = 0$.
        \item Assume $\vX_t$ is an arbitrary configuration where $\Phi_t = 0$.
            Then w.h.p.\ there is a $\tau \leq c \cdot \cliveness \cdot n^2 \log n$ such that either $\vX_{t + \tau}$ is stable ($\in \CL$) or contains a resetting agent.
    \end{enumerate}
\end{lemma}

\subsubsection{Proof of \cref{lem:main_states_reset_quickly}, Part~1: Potential Drops Quickly}
\label{sec:proof:lem:all_main_state_config_dies_quickly}

For our proof, we need the following definition of \emph{good time steps}; we show below that the expected value of the potential will decay geometrically by a factor of $1 - \Omega(n^{-2})$ in good time steps.
Recall that a waiting agent $v$ ($\WaitCount v \neq \bot$) is a ``leader'' which is currently (supposed to be) waiting out a phase transition,
    and that above
        we called a pair of agents $u\neq v$ \emph{rank-assigning} when $\Rank t u \neq \bot$, $\Phase t v \neq \bot$, and $\RankT t u \leq \floor{n \cdot 2^{-\PhaseT t v}}$.

\newcommand{\GoodStep}[1]{{\mathcal{G}_{#1}}}

\begin{definition}\label{def:good_steps}
    A time step $t$ is \emph{good} if $\vX_t \in \CM$ and one of the following statements holds.
    \begin{enumerate}
        \item There is a duplicate rank, i.e., there are agents $u \neq v$ such that $\RankT t u = \RankT t v \neq \bot$.
        \item There are two (or more) waiting agents, i.e., there are agents $u \neq v$ such that $\WaitCountT t u \neq \bot$ and $\WaitCountT t v \neq \bot$.
        \item There is a pair of agents where upon interacting, one of the agent's phases will increase. I.e., there are agents $u, v$ with $\PhaseT t v \neq \bot$, and either $\PhaseT t v < \PhaseT t u \neq \bot$, or $\RankT t u = f_k$ for a $k > \PhaseT t v$.
        \item There is a rank-assigning pair; and for all agents $u, v$ with $\PhaseT t u \neq \bot$ and $\PhaseT t v \neq \bot$, we have $\PhaseT t u = \PhaseT t v$;
            and at least a quarter of all phase agents have $\Coin v = 1$.
    \end{enumerate}
    We write $\GoodStep{t}$ for the event that $t$ is a good time step.
\end{definition}

Recall that we assume that $\vX_t \in \CM$.
Then since $\Phi$ takes non-negative integers as values,
    and by Markov's inequality, we have
\begin{align*}\Prob{\bigvee_{\tau \in [c \cdot n^2 \log{n}]} \Phi_{t+\tau} = 0} &= \Prob{\min_{\tau \in [c \cdot n^2 \log{n}]} \Phi_{t+\tau} = 0} \leq \Exp{\min_{\tau \in [c \cdot n^2 \log n]} \Phi_{t+\tau}}.\end{align*}
So it is sufficient to show that this last expected value is in $O(n^{-1})$.

First, we show that the potential $\Phi_t$ exhibits a multiplicative drop in expectation whenever a round is good.
\begin{lemma}\label{lem:potential_onestep_drop}
For any $t'$, $0 < \phi \in \N$,
    \[\ExpCond{\Phi_{t' + 1}}{\Phi_{t'} = \phi, \GoodStep{t'}} \leq \bc{1 - \frac{1}{4n^2}} \cdot \phi.\]
\end{lemma}

\begin{proof}
    First, note that since we assume that the time step is good,
        we are in a main configuration,
        so that no phase agent can decrease its phase,
        and no non-phase agent can become a phase agent.
    And as a consequence,
        we have $\Phi_{t' + 1} \leq \Phi_{t'}$.

    We proceed by case distinction over the four alternatives by which a round can be good,
        and show that in each case, the potential decreases by a factor $1 - 1/(4n^2)$ in expectation.

    \emph{Cases 1 and 2.}
    When there is a duplicate rank at time $t'$,
        there is at least a $\frac{1}{n(n-1)} \geq n^{-2}$ chance of two agents with the same rank interacting, in which case they will initiate a reset
        and the potential drops to $0$.
    The same holds when there are two (or more) waiting agents.
    So conditioning on time $t'$ being good for either of those reasons,
        the expected value of $\Phi_{t'}$ is at most
        \[\frac{1}{n^2} \cdot 0 + \bc{1 - \frac{1}{n^2}} \cdot \phi \leq \bc{1 - \frac{1}{4 n^2}} \cdot \phi.\]

    \emph{Case 3.}
    In this case, there is at least one pair of agents where upon interacting,
        one of the agents will increase its phase; w.l.o.g., we may assume that this is an agent having the minimum currently saved phase.
    Let $s$ be the number of phase agents,
        and let $\ell \geq 1$ be the number of agents having the minimum currently saved phase.
    Now if $\ell \geq s/2$,
        the probability of one of the $\ell$ agents increasing its phase is at least $\ell / (n(n-1)) \geq s / (2n^2)$ (since if one of those $\ell$ agents can increase its phase in an interaction, all of them can).
    Otherwise, if $\ell < s/2$, there are at least $s/2$ agents saving a phase $\geq \ell$,
        so there is also at least an $\ell \cdot s/2 / (n(n-1)) \geq s/(2n^2)$ chance of an interactions where on of the $\ell$ agents will increase its phase.
    Since the potential contribution of an agent is decreasing in its phase,
        the $\ell$ agents each have an above-average contribution to the potential,
        and their potential contribution will drop by at least one half when their phase increases.
    Hence, the expected value of $\Phi_{t'}$ in this case is at most
        \begin{align*}\frac{s}{2n^2} \cdot \phi \cdot \bc{1 - \frac{1}{2s}} + \bc{1 - \frac{s}{2n^2}} \cdot \phi
            &= \phi - \phi \cdot \frac{s}{4sn^2} = \bc{1 - \frac{1}{4n^2}} \cdot \phi.\end{align*}

    \emph{Case 4.}
    The final case is that in which $t'$ is good because there is a rank-assigning pair, no agent can increment its phase (and hence all saved phases are equal),
        and at least a quarter of phase agents has $\CoinT {t'} v = 1$.
    Since all $s$ phase agents save the same phase and there is a rank-assigning pair $u, v$,
        there must in fact be at least $s$ such pairs.
    Assume w.l.o.g.\ that $\RankT {t'} u \neq \bot$ and $\PhaseT {t'} v \neq \bot$;
        since all saved phases are equal,
            it holds for all $v'$ with $\PhaseT {t'} {v'} \neq \bot$ that $\PhaseT {t'} {v'} = \PhaseT {t'} v$, and hence $u, v'$ is also a rank-assigning pair.
    Since at least a quarter of phase agents has $\CoinT {t'} v = 1$,
        we know that there are at least $s/4$ pairs of agents which, when interacting,
        would lead to a phase agent to become ranked (and hence no longer a phase agent).
    So conditioning on this case,
        and using the fact that as all saved phases are equal, each phase agent has equal contribution to the potential,
        the expected value of $\Phi_{t'}$ is at most
        \begin{align*}\MoveEqLeft\frac{s}{4n(n-1)} \cdot \phi \cdot \bc{1 - \frac{1}{s}} + \bc{1 - \frac{s}{4n(n-1)}} \cdot \phi
            = \phi - \phi \cdot \frac{s}{4sn(n-1)} \leq \bc{1 - \frac{1}{4n^2}} \phi.\end{align*}
    Since we have seen that this holds in all four cases,
        we are done.
\end{proof}

To see that the expected minimum value of the potential over $O(n^2 \log n)$ rounds is small,
    we need a lower bound on the number of good time steps in such an interval.

To that end,
    the following lemma considers the number $B_t$ of time steps which are \emph{not} good in the time interval starting with $t$ and ending when $\Phi = 0$.
To be precise,
    letting $T = \min\{t' \geq t \mid \Phi_{t'} = 0\}$ be the next time after $t$ where the potential is $0$,
    we let \[B_t = \sum_{t' = t}^{T - 1} \boldone\{\textup{$t'$ is good}\}.\]
We show $B_t = O(n^2 \log n)$ w.h.p.;
    hence, to ensure that there are $\Omega(n^2 \log n)$ \emph{good} time steps in expectation, a time interval of $O(n^2 \log n)$ will indeed suffice.

\begin{lemma}\label{lem:few_bad_steps}
There is a sufficiently large $c > 0$ and a sufficiently small $0 < c' < 1$ such that such that for any $t$,
    \(\Prob{B_t \leq c \cdot (1 + \cwait) \cdot n^2 \log n + (1-c') (T-t)} \geq 1 - O(n^{-2}).\)
\end{lemma}

\begin{proof}[Proof sketch]
    Assume that a time step $\tau \in [t, T)$ is \emph{not} good.
        Since $\tau \in [t, T)$, we know that the protocol has not reset in the interval $[t, \tau)$.
        Since furthermore $\tau$ is not good, there is at most one waiting agent, there are no duplicate ranks, no interaction increasing a phase agent's saved phase,
            and one of the following is true:\begin{itemize}
            \item either less than a quarter of phase agents' coins show heads,
            \item or there is \emph{no} rank-assigning pair;
            but since there is still a productive pair (as $\Phi_\tau > 0$ since $\tau < T$), there must be a single waiting agent (since there cannot be two).
        \end{itemize}
        For both cases,
            we need to bound the expected number of time steps where the two cases hold separately.

    Let us consider the first case.
    We divide the time steps after time $t$ into non-overlapping epochs of $2n$ time steps each.
    By \cref{lem:coin_heads}, for each of these time steps,
        if the protocol is in $\CM$ at the beginning of the epoch (which it will be unless a reset is triggered by \cref{obs:cmain_stays_unless_reset}),
        there is a constant probability $p$ of either there being a reset in the epoch or there being at least a quarter of phase agents' coins showing heads
        for at least a constant fraction $c$ of the epoch's time steps.
    Hence, when considering at least $c' n$ epochs for some large constant $c'$ depending on $p$,
        a Chernoff bound guarantees that
        either there is a reset in these epochs
        or for at least a constant fraction of time steps during the epochs,
        at least a quarter of phase agents' coins showing heads.
    Since we do not make any guarantees about there possibly being less than $c n^2 \log n$ non-good time steps for a sufficiently large $c$,
        the assumption on the number of epochs is safe.

    Now for the second case.
    Let $t$ be a time where this case occurs, but where it hasn't occurred in the previous time step.
    We bound the number of steps until the unique waiting agent transitions out of a waiting state (or a reset is triggered).
    Until this happens, besides a reset triggering,
        the only change that can occur (besides liveness checker values or coins changing) is that the waiting agent decreases its counter by one in an interaction with a phase agent.
    Since there is no phase-increasing interaction,
        all phase agents must be saving the same phase, let this be $k$.
    Since there is no productive pair,
        there is no agent $u$ with $\bot \neq \RankT t u \leq \floor{n \cdot 2^{-k}}$;
        and since there are no duplicate ranks, there can be at most $n - \floor{n \cdot 2^{-k}}$ ranked agents.
    So there must be at least $\floor{n \cdot 2^{-k}} - 1$ phase agents (because there is one waiting agent as well);
    However, there is at least one phase agent, since otherwise, there would be no productive pair and $\Phi_t = 0$.
    So the time until the waiting agent transitions out of a waiting state (or there is a reset)
        is stochastically dominated by a negative binomial random variable
        $\NegBin(r, p)$ with $r = \ceil{\cwait \cdot \log n}$ and $p = \max\{1, \floor{n \cdot 2^{-k}} - 1\} / n^2)$.
    By \cref{lem:negbin_tailbound},
        with probability at least $1 - 2n^{-\gamma}$ this is at most
        \[\frac{2}{p} \cdot \bc{r + 2\log n} \leq \frac{n^2}{\max\{1, n \cdot 2^{-k} - 1\}} \cdot \bc{\ceil{\cwait \log n} + 2 \log n}
            \leq c (\cwait + 1) n 2^k \log n\]
            for sufficiently large $n$ and some constant $c$.
    Now if this case occurs at most once for each possible $k \in [\ceil{\log_2 n}]$,
        this would give a total time of $c(\cwait + 1) n \log n \sum_{k=1}^{\ceil{\log_2 n}}2^k$,
        and with the sum being at most $4n$ (see the proof of \cref{lemmacorrectness_non_ss} in \cref{apx:sec:unstable_ranking}),
        we have the claimed time.
    Otherwise,
        consider the first time that the case occurs a second time for some value of $k$.
    Then w.h.p.\ after the time bound above,
        the waiting agent reaches rank $1$, and at that time all phase agents will save phase $k$, meaning that formerly waiting agent is an unaware leader,
            and will assign rank $f_{k+1} + 1$ on the next interaction with a phase agent.
    But this rank will already have been assigned after the first time this case was encountered for this value of $k$.
    And since all saved phases are at least $k$ and cannot decrease without going through a reset, this rank persisted, and there are now two agents having the same rank.
    So within at most $\Geom(1/(n(n-1)))$ rounds, these agents will interact and trigger a reset (if a reset doesn't occur before then).
    As this is in $O(n^2 \log n)$ w.h.p., we are done.

    Since these two cases cover all relevant bad time steps, this proves the claim.
\end{proof}

We are now ready to show that $\Exp{\min_{\tau \in [c n^2 \log n]} \Phi_{t+\tau}} = O(n^{-1})$, as required.

Let $T_{G,i}$ be the $i$th time step after $t$ which is good.
By the law of total expectation,
Write $\Phi_{t, t'} \coloneq \min_{t \leq \tau \leq t'} \Phi_{\tau}$ for the minimum of $\Phi_\tau$ over the time interval $[t, t']$.
Clearly, $\Phi_{t, t'}$ is monotonically non-increasing in $t'$, and non-negative.
Furthermore, note that $\Phi_{t'+1} > \Phi_{t'}$ iff $\Phi_{t,t'} = 0$, because the protocol was in a main configuration at time $t$,
    and for the potential to increase, a non-phase agent needs to become a phase agent
    or a phase agent must decrease its phase, which can only happen after a reset.
Hence, also, $\Phi_{t'} > 0$ if and only if $\Phi_{t,t'} > 0$, and in that case, $\Phi_{t'} = \Phi_{t, t'}$.
So
\begin{align*}
\Exp{\Phi_{t, T_{G,i}+1}}
   &\leq \sum_{\phi > 0} \Prob{\Phi_{T_{G,i}}=\phi}  \cdot \left(\ProbCond{\GoodStep{T_{G,i}}}{\Phi_{T_{G,i}} = \phi} \cdot \bc{1 - \frac{1}{4n^2}} \phi + \ProbCond{\neg \GoodStep{T_{G,i}}}{\Phi_{T_{G,i}} = \phi} \cdot \phi\right)
\\ &= \sum_{\phi > 0} \Prob{\Phi_{T_{G,i}}=\phi} \bc{1 - \frac{1}{4n^2}} \phi
= \bc{1 - \frac{1}{4n^2}} \cdot \sum_{\phi \geq 0} \Prob{\Phi_{t,T_{G,i}}=\phi} \phi
= \bc{1 - \frac{1}{4n^2}} \Exp{\Phi_{t,T_{G,i}}}.
\end{align*}
Applying this repeatedly and using the monotonicity of $\Phi_{t,t'}$ in $t'$, we obtain
\[\Exp{\Phi_{t,T_{G,i}+1}} \leq \bc{1 - \frac{1}{4n^2}}^i \Exp{\Phi_{t,t}}
    \leq \exp\bc{-\,\frac{i}{4n^2}} \cdot n,\]
    so that
\[\Exp{\Phi_{t,T_{G,8n^2 \log n}}} \leq \exp\bc{-2 \log n} \cdot n = n^{-1}.\]
Finally, \cref{lem:few_bad_steps} implies that
    \[\Prob{T_{G,6n^2 \log n} \leq t + (8 + c \cdot \cwait)n^2 \log n} \geq 1 - n^{-2}.\]
So that indeed
    \[\Prob{\min_{\tau \in [(8 + c \cdot \cwait)n^2 \log n]} \Phi_{t+\tau} > 0}
        \leq n^{-1} + n^{-2}.\]

\subsubsection{Proof of \cref{lem:main_states_reset_quickly}, Part~2}
\label{sec:proof:lem:dead_configs_end_quickly}

Recall that we assume that $\Phi_{t} = 0$.
So by definition of $\Phi$,
    either the protocol is resetting at time~$t$,
    or there are no productive pairs at time~$t$.
If the protocol is resetting at time $t$, the proof of this part of  \cref{lem:main_states_reset_quickly} is already done with $\tau = 0$,
    so only the latter case remains.
In the following we call a non-legal configuration with no productive pairs a \emph{dead configuration}.

\begin{observation}
\label{observation_dead}
    Assume the protocol is in a dead configuration at time $t$.
    Then from time $t$ until the protocol resets (if ever), all of the following hold:\begin{enumerate}
        \item All waiting agents remain waiting,
        \item no ranked agent changes its rank, and
        \item there are no productive pairs, i.e., the configuration remains dead.
    \end{enumerate}
\end{observation}

\begin{lemma}\label{lem:reset_quickly_when_dead}
    Assume that at time $t$, there are no productive pairs
        and that the configuration is not stable,
        i.e., we are in a dead configuration.
    Then w.h.p.\ within $O(n^2\log n)$ interactions, the protocol will reset.
\end{lemma}
The proof of this lemma follows from the following lemmas.

\begin{lemma}\label{lem:reset_quickly_when_dead_and_dupes_exist}
    Let $c$ be a sufficiently large constant.
    Assume the protocol is in a main configuration without productive pairs, but with duplicate ranks, at time $t$.
    Then w.h.p.\ there is a time $\tau \in [t, t + c \cdot n^2 \log n]$ such that the protocol resets at time $\tau$.
\end{lemma}

\begin{proof}
    According to Observation \ref{observation_dead}, none of the ranked agents will change its state unless the protocol resets. Let $u$ and $v$ be two agents with the same rank. According to Protocol \ref{alg:stable_unsafe_silent} if $u$ and $v$ are selected for interaction, then the protocol resets (unless there was a reset between step $t$ and the time step in which $u$ and $v$ meet). We know that in each time step $u$ and $v$ are selected with probability  $2/(n(n-1))$, independently of the agents selected for interaction in any other step. Clearly, $u$ and $v$ do not interact in $c \cdot n^2 \log n$ time steps with probability
    \[
     \left(1-\frac{2}{n(n-1)} \right)^{c \cdot n^2 \log n} < \frac{1}{n^2}
    \]
    whenever $c$ is large enough. This implies that w.h.p.\ the protocol resets in some time step $\tau \in [t, t + c \cdot n^2 \log n]$.
\end{proof}

\begin{lemma}\label{lem:reset_quickly_when_single_rankless}
    Let $c$ be a sufficiently large constant.
    Assume the protocol is in a main configuration without productive pairs, and with a single agent without a rank, at time $t$.
    Then w.h.p.\ there is a time $\tau \in [t, t + c \cdot \cliveness \cdot n^2 \log n]$ such that the protocol resets at time $\tau$.
\end{lemma}

\begin{proof}
    Let $u$ be the agent without a rank. According to Observation \ref{observation_dead}, $\WaitCount u \neq \bot$ or $\Phase u \neq \bot$ in all steps $\tau \geq t$ until the protocol resets.

    As we only have one agent without a rank, there must be two agents with the same rank, or at least one of the ranks $n-1$ or $n$ are assigned at time $t$. If there are two agents with the same rank, \cref{lem:reset_quickly_when_dead_and_dupes_exist} implies that w.h.p.\ there is a time $\tau \in [t, t + c \cdot n^2 \log n]$ such that the protocol resets at time $\tau$.
    Otherwise, agent $u$ decrements $\LivenessCount u$ every time step in which it interacts with an agent with rank $n-1$ or $n$.
    We know that $u$ interacts with such an agent in a time step with probability at least $2/(n(n-1))$, independently of any other time step. Thus, applying Chernoff bounds \cite{DBLP:journals/ipl/HagerupR90}, we obtain that with probability at least $1-n^{-2}$, $\cliveness \cdot \log n+1$ such interactions will occur within $c \cdot \cliveness \cdot n^2 \log n$ steps, if $c$ is large enough ($\creset$ is the constant specified in \cref{alg:stable_unsafe_silent}). This implies that within $c \cdot \cliveness \cdot n^2 \log n$ time steps, $\LivenessCount u$ reaches $0$ w.h.p., leading to a reset according to \cref{alg:stable_unsafe_silent}.
\end{proof}

\begin{lemma}\label{lem:reset_quickly_when_multiple_rankless}
    Let $c$ be a sufficiently large constant.
    Assume the protocol is in a main configuration without productive pairs, and with two or more agents without a rank, at time $t$.
    Then w.h.p.\ there is a time $\tau \in [t, t + c \cdot \cliveness \cdot n^2 \log n]$ such that the protocol resets at time $\tau$.
\end{lemma}
\begin{proof}
We adapt a proof from \cite[Lemma 3.3]{DBLP:conf/podc/BurmanCCDNSX21}, which, in turn, is adapted from \cite[Lemma 1]{DBLP:conf/dna/AlistarhDKSU17}.

    For an agent $v$ define $C_\tau(v)$ to be $-\infty$ if $\LivenessCount{v} = \bot$ at time $\tau$ or if the protocol resets at a time between $t$ and $\tau$;
        otherwise let $C_\tau(v)$ be the value of $\LivenessCount{v}$ at time $\tau$.
    Furthermore, define $\Gamma_\tau(v) = 3^{C_{\tau}(v)}$ (thus $\Gamma_\tau(v) = 0$ if $C_\tau(v) = -\infty$),
        and let $\Gamma_t = \sum_{v \in V} \Gamma_t(v)$.
    As long as there are no productive pairs,
        the only way the values $C_\tau(v)$ (and thus $\Gamma_\tau$) can change is
            if there is a reset in step $\tau$ (in which case $\Gamma_{\tau+1}$ becomes $0$)
            or if two agents with $\LivenessCount \cdot \neq \bot$ (i.e., two unranked agents) interact in step $\tau$.

    Now let $k \geq 2$ be the number of agents without a rank at time $t$. Due to \cref{observation_dead} all these agents remain unranked until the protocol resets.
    The probability that two agents without a rank, say $v$ and $u$, interact at a given time (before a reset)
        is $\frac{k(k-1)}{n(n-1)}$.
    In such an interaction,
        they both reset their $\texttt{aliveCount}$s to the maximum of the  $\texttt{aliveCount}$s of $u$ and $v$, minus one.
    Then,
        \[\Gamma_{\tau+1}(u) + \Gamma_{\tau+1}(v) = 2 \cdot 3^{\max\{C_\tau(u), C_\tau(v)\}-1} \leq \frac{2}{3} \max\{3^{C_\tau(u)}, 3^{C_\tau(v)}\} \leq \frac{2}{3} \cdot (\Gamma_\tau(u) + \Gamma_\tau(v)).\]
    Conditioned on the event that in some step $\tau$ two unranked agents are chosen for interaction, for two arbitrary but fixed unranked agents $u,v$ we have
        $\Gamma_{\tau+1} \leq \Gamma_\tau -
        (\Gamma_\tau(u) + \Gamma_\tau(v))/3$
        with probability $2/(k(k-1))$. Then,
        \[\Exp{\Gamma_{\tau+1}} \leq \bc{1 - \frac{(k-1)k}{n(n-1)}} \cdot \Exp{\Gamma_\tau} + \frac{(k-1)k}{n(n-1)} \cdot \bc{\Exp{\Gamma_\tau} - \frac{2}{3k} \cdot \Exp{\Gamma_\tau}}
           \leq  \bc{1 - \frac{2k-2}{3n^2}} \cdot \Exp{\Gamma_\tau}.\]
    The value of $\Gamma$ is always at most $n \cdot n^{2\cliveness}$,
        so in $c \cdot \cliveness n^2 / (k-1) \cdot \log n \leq c \cdot \cliveness n^2 \log n$ time steps,
        the expected value of $\Gamma$ will decrease below $n^{-c'}$ for any predefined constant $c'$ if $c$ is large enough. Applying now Markov's inequality, we obtain the lemma.
\end{proof}

\subsection{Proof of \cref{lem:recovery}}
\label{sec:proof_recovery}

Recall that we need to show that when $\vX_t \in \CT$ is a triggered configuration,
    the protocol will enter a leader-electing configuration ($\in \CLE$, see \cref{sec:proof_leader_elect} below) within $O(n \log n)$ interactions.

The following lemma, which describes the behavior of $\ResetProt$,
    uses the notion of an \emph{awakening configuration},
    which is the first partially computing configuration reachable from a fully dormant configuration.

\begin{lemma}[Corollary of Theorem 3.4 in \cite{DBLP:journals/corr/abs-1907-06068}]
    Let $\Rmax = 60 \ln n$ and $\Dmax = \Omega(\log n + \Rmax)$.
    Starting from a triggered configuration\footnote{Note that this is called a ``partially-triggered configuration in \cite{DBLP:journals/corr/abs-1907-06068}.},
        we reach an awakening configuration in $\Theta(\Dmax n)$ interactions\footnote{In \cite[Theorem 3.4]{DBLP:journals/corr/abs-1907-06068}, this is just an upper bound. However, as the $\texttt{delayCount}$ of some agent has to decrease from $\Dmax$ to $0$, and $\Dmax = \Omega(\log n)$ with a sufficiently large constant, it indeed also takes at least $\Omega(n \Dmax)$ interactions for this to occur.} with probability at least $1 - O(1/n)$.
\end{lemma}

Now recall that even while $\PropagateReset$ is running,
    responding agents flip their coin on every interaction,
    and that our definition of $\CLE$
    requires that the difference in numbers between coins showing $1$ and $0$ is at most $\frac{n}{4\log n}$.
As $\PropagateReset$ takes at least $\Omega(n \Dmax)$ interactions,
    and $\Dmax = \Omega(\log n)$ with the leading constant being our choice,
    the following lemma shows that this coin property indeed holds with sufficient probability.
As the proof of the Lemma is entirely analogous to that found in \cite{DBLP:conf/soda/BerenbrinkKKO18} (replacing occurrences of $\log \log \log n$ with $\log \log n$), we omit it here.

\begin{lemma}[cf.\ Lemma 3 in \cite{DBLP:conf/soda/BerenbrinkKKO18}]
    Let $\gamma > 0$ and consider an interaction $t$ with $n \log(4 \log n) / 2 \leq t \leq n^\gamma$.
    Then the number of coins that equal zero at the beginning of interaction $t$ lies with probability at least $1 - n^{-\gamma}$ in $(1 \pm 1 / (4 \log n)) \cdot n / 2$.
\end{lemma}

\subsection{Proof of \cref{lemma:convergence_leader_elect}}
\label{sec:proof_leader_elect}

We start from a configuration that results from executing $\ResetProt$ protocol.
Intuitively, these are all configurations the population is ready to execute \FastLeaderElect.
To be precise,  all agents are either dormant (and wait to start the protocol) \emph{or} are in the initial state of the leader election.
Furthermore, the values of the flip bits have converged to a distribution such that roughly half of all agents have either bit.
In other words, all agents have been dormant long enough for the bit to settle, and first, the agent has just woken up from being dormant.
Formally, we define these configurations as follows.

\begin{definition}[\CLE]
In a leader electing configuration $\vX \in \CLE$, all agents are either dormant, i.e., it holds  $\DelayCount{v} \geq 0$, or are in an initial state $q_{0,i}$ ($i \in \{0,1\}$) for $\FastLeaderElect$, i.e., their variables have the following values \begin{align*}(\LeaderBit{v},\LeaderDone{v},\CoinCount{v}, \LECount{v}, \Coin{v}) = (0,0,\lceil\log n\rceil,\Lmax, i).\end{align*}
Furthermore, the following (global) property holds for the agents' coins:
    \begin{align*}
        \big||\{ v \in V \mid \Coin{v} = 1 \}| - |\{ v \in V \mid \Coin{v} = 0 \}| \big|  \leq \frac{n}{4\log n}.
    \end{align*}
\end{definition}

Starting from such a configuration, the population will reach configuration with a \emph{unique} leader with constant probability.
\begin{lemma}[Prob. for Unique Leader]
\label{lemma:leader_prob}
Suppose that the following (global) property holds for the agents' coins:
\begin{align*}
        \big||\{ v \in V \mid \Coin{v} = 1 \}| - |\{ v \in V \mid \Coin{v} = 0 \}| \big| \leq \frac{n}{4\log n}.
\end{align*}
Then, with probability greater than $\nicefrac{1}{8e}$, there is exactly one agent that sets $\LeaderBit{v} = 1$.
\end{lemma}
\begin{proof}
Fix an agent $v \in V$ and let $L_v \in \{0,1\}$ be the indicator that $v$ is a leader and $U_v \in \{0,1\}$ be the indicator
that $v$ is the unique leader.
Note that $\log n \leq \lceil\log n \rceil \leq \log 2n \leq 2\log n$.
Then, the probability that $v$ becomes a leader is lower bounded by
\begin{align*}
    \Prob{L_v=1} &\geq \left(\left(1-\frac{1}{4\log n}\right)\frac{1}{2}\right)^{\lceil\log n \rceil}
     \geq \left(1-\frac{1}{4\log n}\right)^{\lceil\log n \rceil} \frac{1}{2^{\lceil\log n \rceil}} \geq  \left(1-\frac{1}{4\log n}\right)^{\lceil\log n \rceil}\frac{1}{2n} \geq  \left(1-\frac{2\log n}{4\log n}\right)\frac{1}{2n} \\& \geq \frac{1}{4n}.
\end{align*}
Furthermore, using an analogous calculation, it is upper bounded by
\begin{align*}
    \Prob{L_v=1} &\leq \left(\left(1+\frac{1}{4\log n}\right)\frac{1}{2}\right)^{\lceil\log n \rceil}
     \leq \left(1+\frac{1}{4\log n}\right)^{\lceil\log n \rceil} \frac{1}{2^{\lceil\log n \rceil}}
    \leq  \left(1+\frac{1}{4\log n}\right)^{2\log n}\frac{1}{n}
    \\ &\leq  \left(1+\frac{1}{4\log n}\right)^{4\log n \cdot \frac{1}{2}}\frac{1}{n}
    \leq \frac{\sqrt{e}}{n} \leq \frac{2}{n}.
\end{align*}
The first inequality in the last line follows from the fact that $(1+\frac{x}{n})^n \leq e^x$.
Note that these bounds are the same for all agents.
Thus, the probability that $v$ is the \emph{unique} leader is lower can be bounded as follows
\begin{align*}
    \Prob{U_v=1} &= \Prob{L_v = 1} \prod_{w \in V\setminus\{v\}}\left(1-\Prob{L_w = 1}\right)m \geq  \frac{1}{4n} \left(1 - \frac{2}{n}\right)^{n-1}
    \geq  \frac{1}{4n} \left(1 - \frac{2}{n}\right)^{n}  \geq  \frac{1}{4n} \left(e^{-1} \left(1-\frac{4}{n}\right)\right) \geq \frac{1}{8en}.
\end{align*}
Here, we used that $(1-\frac{x}{n})^{n} \geq e^{-x} (1-\frac{x^2}{n})$ and $1 - \frac{4}{n} \geq \frac{1}{2}$ for $n \geq 8$.
As all event $U_v$ are disjoint, we can sum them up and get
\begin{align*}
   \Prob{\sum_{v \in V}{U_v} = 1} = \sum_{v \in V} \Prob{U_v = 1} \geq \sum_{v \in V} \frac{1}{8e}\cdot\frac{1}{n} =  \frac{1}{8e}.
\end{align*}
Thus, the lemma follows.
\end{proof}

Recall that, we trigger a reset if any agent is activated $\Lmax$ times.
If we choose $\Lmax \in \Theta(\log n)$ large enough, the leader will be elected before any agent interacts $({1}/{8})\cdot\Lmax$ times.
Thus, no agent \emph{accidentally} starts with executing the main protocol before there is a unique leader.

Formally, we define these configurations as follows.
\begin{definition}\label{def:csrplus}
$\CSRPlus$ is the set of configurations where all agents' variables have the following values.
\begin{align*}
\LECount{v} &\geq ({7}/{8}) \cdot \Lmax& \text{Counter was not decreased too much.}\\
\LeaderDone{v} &= 1 &\text{No (additional) agent will be leader.}
\end{align*}
Furthermore, there is \emph{exactly one} agent $l \in V$ with $\LeaderBit{l} = 1$, and for \emph{all other} agents $w \in V \setminus \{l\}$, it holds $\LeaderBit{w} = 0$.
\end{definition}
Using standard arguments, we can show the following lemma.
\begin{lemma}
\label{lemma:LE_to_SAFE}
Let $c$ be a sufficiently large constant, and assume $\vX_t \in \CLE$. Then, with constant probability, there is a $\tau \in [c \cdot n^2 \cdot \log n]$ such that  $\vX_{t+\tau} \in \CSRPlus$.
Otherwise, with probability ${1}/{2^d}$, there is a $\tau' \in [d \cdot c \cdot \cdot n^2 ]$ such that  $\vX_{t+\tau'} \in \CLE$.
\end{lemma}

\begin{proof}
We begin with the first two properties and show that there is time step where all agents have $\LECount{v} \geq ({7}/{8}) \cdot \Lmax$ and $\LeaderDone{v}=1$.
Recall that in a configuration from $\CLE$, there is precisely one agent in a state $q_{0,i} \in \QLE$, and all others are dormant.
In each interaction between a dormant agent and an agent with a state from $\QLE$, the dormant agent will always switch to a state from $\QLE$.
Furthermore, whenever two dormant agents interact, they either remain dormant (if it holds $\DelayCount{v} > 0$ for both of them) or they switch to an initial state $q_{0,i} \in \QLE$.
Thus, the time in which all agents switch to a state from $\QLE$ is upper bounded by the time of one-way epidemic, namely $4\gamma\cdot n\cdot \log n$ interactions with probability $1-O(\nicefrac{1}{n^\gamma})$.
Furthermore, note that each agent $v \in V$ switches from being dormant sets $\LECount{v} = \Lmax$ upon initialization.

Let $t_1$ be the time step where the last agent $v \in V$ switches to $q_{0,i} \in \QLE$.
From this point on, each interaction between $v$ and $w$, will decrease $\CoinCount{v}$ unless $\LeaderDone{v}=1$.
Once, it holds $\CoinCount{v} = 0$, it also sets $\LeaderDone{v}=1$.
After $(1+\gamma)\cdot n\log n$ steps each agent has been activated $\lceil\log n\rceil$ times with probability $1-n^{-\gamma}$ and it holds $\LeaderDone{v}=1$ for all $v \in V$ .
Thus, after $t_1 \leq 4\gamma\cdot n\cdot \log n$ global steps for waking up and $t_2 \leq (1+\gamma)\cdot n\log n$ global steps for flipping coins, all agents have $\LeaderDone{v} = 1$.

During all this time, each agent was activated $t_1 + t_2 \leq (5\gamma+1)\cdot n \cdot \log n$ times on expectation.
For a large enough choice of $\Lmax > 100\gamma\log n$, no agent was activated less than $(\nicefrac{1}{8})\cdot\Lmax$ times with probability $1-O(n^{-\gamma})$.
This follows from a straightforward application of the Chernoff bound \cite{DBLP:journals/ipl/HagerupR90}.

Now, we get to the second property, the number of leaders.
By \cref{lemma:leader_prob} we know that exactly one leader is elected with constant probability.
In this case, we are done.
Therefore, it remains to show that we reset if have no or more than one leader.

\emph{Case 1: No leader.} \quad If no leader is elected after $\Theta(\Lmax \cdot n)$ interactions, one agent $v \in V$ must have been activated $\Lmax$ times and it holds $\LECount{v} = 0$.
As this agent is not a leader, it will trigger a reset.

\emph{Case 2: Two or more leaders.} \quad
Suppose we have two or more leaders.
In the following we assume w.l.o.g.\ that at least two leaders start waiting. Otherwise, we are in Case $1$ or the protocol \emph{overwrites} the other leaders.
We will show that in this case, either two leaders assign rank $r_1 = \lceil\frac{n}{2}\rceil + 1$ or trigger a reset in $O(n^2)$ steps.

First, we argue that two leaders enter the waiting state at most $O(n\log n)$ steps apart, w.h.p.
Suppose that the first leader starts to wait in some step $t_1$.
Then, the epidemic that turns any other leader into a phase agent reaches every other agent $O(n \log n)$ steps later.
Thus, if a second leader also starts waiting in some step $t_2 \geq t_1$, we can assume $|t_2 - t_1| \in O(n\log n)$.

Let now $\ell_1$ and $\ell_2$ be the first two leaders that stop waiting and start ranking.
We claim that both agents stop waiting while all other unranked agents are still in the first phase.
This is trivially true for $\ell_1$.
To prove this for $\ell_2$, note that $\ell_1$ must interact with $n/2$ agents to finish the first ranking phase.
This requires at least $N = \Omega(\frac{n^2}{\log n})$ global time, w.h.p.
If $t_1$ is the time where any leader started waiting, at time $t_1 + N$, all agents are still in the first phase (unless, of course, we triggered a reset).
We argue that $\ell_2$ stops waiting before this time.
During this time, at least half the agents remain unranked (if no other leader wakes up).
Thus, if the second leader has not stopped waiting up until time $t_1 + N$, it would have either interacted with $\Theta(\frac{N}{n}) \in \omega(\Wmax)$ unranked agents or with another waiting leader, w.h.p.
Therefore, it either stopped waiting within $O(n^2)$ steps or triggered a reset.

Thus, if no reset is triggered, the first two leaders both wake up in the first phase and assign $r_1$.
In every round after this, these two agents interact with probability $\frac{1}{n(n-1)}$.
Therefore, the probability that they have not interacted after $\tau$ steps is dominated by the geometric distribution with parameter $\nicefrac{1}{n^2}$.

\end{proof}

Note that this lemma implies that after $O(\log n)$ failed attempts to enter $\CSRPlus$, we succeed.
The time spent in these failed attempts is $O(n^2\log n)$.

\subsection{Proof of \cref{lemma:convergence_good_case}}
\label{sec:proof_good_case}

Recall that $\CL$ is the set of all configurations where all agents have unique rank.
\begin{lemma}
Let $c$ be a sufficiently large constant, and assume $\vX_t \in \CSRPlus$.
Then, \emph{unless we trigger a reset}, there is a $\tau \in [c \cdot n^2 \log n]$ such that  $\vX_{t+\tau} \in \CL$.
\end{lemma}
\begin{proof}
As we have exactly one leader and no more leaders will be elected, this leader will eventually start waiting.
This starts an epidemic that turns all agents into phase-counting agents.
Recall that these agents \emph{forget} their value of $\LECount{v}$.
Since $\LECount{v} \geq (\nicefrac{7}{8})\cdot\Lmax$ for all agents, the epidemic will reach them before they perform $(\nicefrac{7}{8})\cdot\Lmax$ interactions, w.h.p., for a large enough choice of $\Lmax$.
Thus, eventually, one agent is waiting, and all agents are phase-counting.
Given that there is no reset in the next $O(n^2 \log n)$ interactions, the lemma follows from \cref{lemmacorrectness_non_ss}.
\end{proof}

Note that this lemma only holds under the premise that during execution, we never trigger a reset.
So, in the remainder of the proof, we show that we do not reset w.h.p.
We must consider all rules that could potentially trigger a reset when starting in an arbitrary configuration in $\CSRPlus$ and argue why these rules are \emph{not} triggered with sufficient probability.
Recall that a reset can \emph{either} is triggered through two agents with the same label, an emergency reset issued by a waiting agent \emph{or} by a phase counting agent whose liveness counter expires.
We look at these three rules separately.

\paragraph{Rule 1: Reset Through Duplicate Labels}
First, we note that the protocol, if started from $\CSRPlus$ will never assign a label twice, w.h.p.
Therefore, w.h.p., no agent will trigger a reset because of an interaction with the same label.

\paragraph{Rule 2: Reset Through Waiting Agent}
A waiting or phase counting agent triggers a reset if it interacts with an agent of label $n$ or $n-1$ more than $\Lmax$ times before being labeled itself.
This case can be ruled out through an appropriate choice of the tunable variable $\Lmax$.
By \cref{lemmacorrectness_non_ss} (and waiting for the right coins slows it down by a constant factor) the protocol stabilizes in $O(\Wmax\cdot n^2)$ interactions, w.h.p.
Thus, after $O(\Wmax\cdot n^2)$ interactions, every agent has a label unless a reset is triggered.
We first show the following claim.
\begin{claim}
There is a constant $c_1 > 0$ that does not depend on $\Lmax$, s.t., every pair of agents interacts at most $c_1 \cdot \Wmax$ times before all agents are labeled, w.h.p.
\end{claim}
\begin{proof}
Condition on the event that after $O(\Wmax\cdot n^2)$ interactions, every agent has a label. This happens w.h.p.
Recall that the probability of two agents interacting in a given step is $\nicefrac{1}{n(n-1)}$ as we use the uniform scheduler.
Thus, for a given pair $(v,w)$ of agents, the expected number of interactions between $v$ and $w$ within $O(\Wmax\cdot n^2)$ steps is $O(\Wmax)$.
As interactions are independent and can be modeled as binary random variables, we can apply the Chernoff bound \cite{DBLP:journals/ipl/HagerupR90}.
For each pair $v,w$, we can use this well-known bound to show that the probability that $v$ and $w$ interact more than $c_1\cdot\Wmax$ times (where $c_1$ is large constant that depends on the constants hidden in $O(\Wmax)$) during this time is at most $\nicefrac{1}{n^{c_2}}$ where $c_2$ depends on $c_1$.
A union bound over all $n(n-1)$ pairs of agents yields that the probability of \emph{any} pair interacting more than $c_1\cdot\Wmax$ times is less than $\nicefrac{1}{n^{c_2-2}}$.
Thus, an appropriate choice of $c_1$ yields the lemma.
\end{proof}
From this, we can conclude that any agent $v \in V$ interacts with any set of two agents $w_1,w_2 \in V$ at most $2\cdot c_1\cdot\Wmax$ times w.h.p.
By choosing $\Lmax > 2\cdot c_1\cdot\Wmax$, we conclude that, w.h.p., no waiting agent interacts $\Lmax$ times with any two specific agents.
This includes, in particular, the agents with labels $n$ and $n-1$ at any point during the execution.
So, no reset is triggered, w.h.p.

\paragraph{Rule 3: Reset Through Liveness Checker}
Finally, we argue why a phase counting agent will not trigger a reset, w.h.p.
This (arguably) requires the most intricate proof.
To this end, consider a configuration $\vX_t$ and assume w.l.o.g.\ that we have $k$ phase counting agents.
Recall that these agents do \emph{not} have a label and count the phase.
In the following, we denote the set of phase counting agents in configuration $\vX_{t'}$ as $P_{t'} \subset V$.
Note that $P_{t} \subseteq P_{t+1}$ (unless we trigger a reset) and
the number of phase counting agents is monotonically decreasing.
Furthermore, we use $\ell$ to denote the agent that assigns the next label when interacting with an agent from $P$ or is waiting.
Recall from the analysis that there is always \emph{exactly} one such agent in every configuration, w.h.p.
Finally, we will divide time into \emph{phases} of $\tau = 4 \cdot \frac{n^2}{k}$ interactions, s.t., phase $i$ includes all configurations from $\vX_{t+i\tau}$ to $\vX_{t+(i+1)\tau-1}$.
We call $T = c\cdot4 \log n$ continuous phases an \emph{epoch}.
Here $c > 1$ is a tunable constant we will fix in the analysis.
We will show the following lemma.
\begin{lemma}
\label{lemma:epoch_invariant}
Consider an epoch $\vX_t, \ldots, \vX_{t+T\cdot\tau}$ and assume for all agents it holds $ \geq (\nicefrac{7}{8})\cdot\Lmax$ in $\vX_t$.
Then, w.h.p., it holds that
\begin{enumerate}
    \item For all $t' \in [0,T\cdot\tau]$ and all $v \in P_{t + t'}$, it holds $\LiveCount{v} > 0$.
    \item For all $v \in P_{t + T\cdot\tau}$, it holds $\LiveCount{v} \geq (\nicefrac{7}{8})\cdot\Lmax$.
\end{enumerate}
\end{lemma}
Informally, this lemma states that during an epoch, no reset is triggered, and given that all agents have a high counter value in the beginning, they have a high counter value at the end.
Let $\mathcal{E}_i$ be the event that no reset is triggered in the $i$th epoch \emph{and} for all $v \in P$ at the end of the epoch, it holds $\LiveCount{v} \geq (\nicefrac{7}{8})\cdot\Lmax$.
Note that each epoch is of length at least $O(n \log n)$.
Therefore, if the event $\mathcal{E}_i$ holds for $\eta \in O(n)$ consecutive epochs, we do not trigger a reset within $O(n^2 \log n)$ interactions.
Thus, we will show that for some $c > 1$ that
\begin{align}
\label{eqn:epoch_equation}
    \Prob{\bigcap_{i=1}^\eta \mathcal{E}_i} \geq 1 - n^{-c}.
\end{align}
Let $\mathcal{P}_t$ be the event that for all $v \in P_t$, it holds $\LiveCount{v} \geq (\nicefrac{7}{8})\cdot\Lmax$
Furthermore, $S_i = (t,k)$ is the event that the $i$th starts in step $t$ with $k$ phase counting agents.
By \cref{lemma:epoch_invariant}, for \emph{all} possible choices of $t$ and $k$, it holds for some universal $c' > 1$ that
\begin{align*}
     \Prob{\mathcal{E}_i \mid S_i = (t,k) \cap \mathcal{P}_t } \geq 1 - {n^{-c'}}.
\end{align*}
In particular, the lemma holds conditioned on everything else that happened before step $t$ as only $\mathcal{P}_t$ and $S_i$ are relevant for the lemma.
Let now $\mathcal{S}_i$ be the set of all possible realizations of $S_i$.
Then, using the chain rule of conditional probability and the law of total probability, we get
\begin{align*}
\Prob{\bigcap_{i=1}^\eta \mathcal{E}_i} &= \prod_{i=1}^\eta \Prob{\mathcal{E}_i \mid \bigcap_{j=1}^{i-1} \mathcal{E}_j}\\
&= \prod_{i=1}^\eta \sum_{(t,k) \in \mathcal{S}_i} \Prob{S_i = (t,k) \mid \bigcap_{j=1}^{i-1} \mathcal{E}_j} \Prob{\mathcal{E}_i \mid S_i \bigcap_{j=1}^{i-1} \mathcal{E}_j} \\
&= \prod_{i=1}^\eta \sum_{(t,k) \in \mathcal{S}_i} \Prob{S_i = (t,k) \mid \bigcap_{j=1}^{i-1} \mathcal{E}_j} \Prob{\mathcal{E}_i \mid S_i = (t,k) \cap \mathcal{P}_t} \\
&\geq \prod_{i=1}^\eta \sum_{(t,k) \in \mathcal{S}_i} \Prob{S_i = (t,k) \mid \bigcap_{j=1}^{i-1} \mathcal{E}_j} (1-n^{-c'}) \\
&\geq \prod_{i=1}^\eta (1-n^{-c'}) \prod_{i=1}^\eta  \sum_{(t,k) \in \mathcal{S}_i} \Prob{S_i = (t,k) \mid \bigcap_{j=1}^{i-1} \mathcal{E}_j}  \\
&\geq (1-n^{-c'})^\eta \prod_{i=1}^\eta  \sum_{(t,k) \in \mathcal{S}_i} \Prob{S_i = (t,k) \mid \bigcap_{j=1}^{i-1} \mathcal{E}_j}  \\
&\geq (1-n^{-c'})^\eta \prod_{i=1}^\eta  1 \geq (1-n^{-{c'-2}}).
\end{align*}
Thus, inequality \eqref{eqn:epoch_equation} holds for $c = c' + 2$ and the lemma follows.

\begin{proof}[Proof of \cref{lemma:epoch_invariant}]
Before we start with our main argument, let us quickly recall the behavior of these agents.
If they interact with the labeling agent $\ell$, they either get labeled (if their coin is $0$ ) or reset their liveness counter to $\Lmax$ (if their coin is $1$) .
If $\ell$ is waiting, they always reset their liveness counter to $\Lmax$.
Furthermore, whenever two phase-counting agents interact, they agree on the maximum of their respective liveness counters and decrease them by one.
We suppose that the agents perform a slightly different protocol for this proof.
For the sake of argument, assume that agents have infinite memory and act like agents in a message-passing system.
The adapted protocol works as follows:
when agent $\ell$ interacts with an agent $v \in P$ whose coin is $1$, it starts a broadcast $b$.
A broadcast $b$ is a message that contains $l_b$, a counter that stores how often it has been forwarded since its creation.
Initially, the counter is set to $l_b = 0$.
Whenever two agents $v,w \in P$ interact, they forward all broadcasts they know and increase their counters by $1$.
Assume we run both protocols simultaneously, letting the same agents interact in each step.
Call the resulting protocol, the \emph{coupled} protocol.
Then, the connection between broadcast protocol and our protocol is as follows.
\begin{claim}
    Consider a broadcast $b$ that has been forwarded $l_b$ times.
    In the coupled protocol, the liveness counter of an agent $v \in V$ that has received $b$ has a value of at least $\Lmax - l_b$ in the original protocol.
\end{claim}
\begin{proof}
This can be shown through a simple induction on the lifespan of a broadcast $b$.
\begin{enumerate}
    \item For the base case, recall broadcast $b$ is created by interacting with $\ell$.
    Suppose an agent $v$ interacts with agent $\ell$ and creates $b$.
 Its initial count is $l_b = 0$.
 This interaction also sets $v$'s counter to $\Lmax = \Lmax - l_b$.
  Thus, initially, our claim holds.
 \item For the step, suppose that agents $v$ and $w$ interact. W.l.o.g., let $v$ have the higher liveness counter. Then, $v$ must have a broadcast message $b_v$ with counter $l_{b_v}$ such that $\LiveCount{v} \geq \Lmax - l_{b_v}$.
 After the interaction, $v$ has decreased its counter by one and increased $l_{b_v}$ by $1$ and $w$ has the same value as $v$ and also knows $b_v$.
 Thus, the claim follows.
\end{enumerate}
This proves the claim.
\end{proof}
Therefore, we can use the broadcast time to bound the drift of the liveness counters.
We show the following claim.
\begin{claim}
Suppose a broadcast $b$ was received by an agent $v \in V$ within $T'$ phases.
For any choice $T' \geq c\log n$, it holds that
\begin{align}
    \Prob{l_b \leq 200T'} \geq 1-n^{-c}.
\end{align}
\end{claim}
\begin{proof}
Fix a broadcast $b$ started in configuration $\vX_{t_0}$ by some agent $v_1 \in V$ interacting with $\ell$ and is known by agent $v$ in step $t_0 + T' \cdot \tau$.
If $v$ receives the broadcast with value $l_b = l$, we can create a \emph{witness sequence} that \emph{proves} that $v$ has the broadcast $b$ and the current counter is $l$.
Seeking formalization, there must be the sequence $W = \left((v_1, w_1, t_1), \ldots, (v_l, w_l, t_l) \right)$ of time steps $t_1, \ldots, t_{l} \in [T'\cdot\tau]^l$ and pairs $(v_1,w_1), \ldots, (v_l,w_l) \in P^l$.
If $v_i = v_{i+1}$, agent $v_i$ has sent $b$ to $w_{i}$ in step $t_i$ (and thereby increases its counter).
Otherwise, If $ v_{i+1} = w_i$, agent $v_i$ has initially received $b$ from $w_{i+1}$ in step $t_i$ (and thereby increases its counter).
This distinction is necessary because the broadcast counter is increased on every interaction and not only in the step that it is received.
Note that if the broadcast counter is $l$, we can construct a witness sequence of length $l$.
We now bound the probability of such a sequence in general.
First, we are interested in the probability of a fixed witness sequence $W$.
Let $I(v_i,w_i,t_i)$ be the event that $v_i$ interacts with $w_i$ in step $t_i$.
Note that for any pair $v_i,w_i \in V$ and any step $t_i$, it holds regardless of everything that happened before step $t$ that
\begin{align*}
    \Prob{I(v_i, w_i, t_i)} \leq \frac{1}{n(n-1)}.
\end{align*}
Therefore, by the chain of conditional probability
\begin{align*}
\Prob{W} \Prob{\bigcap_{i = 1}^l I(v_i,w_i,t_{i}} = \prod_{i=1}^l \ProbCond{I(v_i,w_{i},t_{i}}{\bigcap_{j < i} I(v_j,v_{i},t_{j}}
\leq \prod_{i=1}^l \frac{1}{n(n-1)} = \left(\frac{1}{n(n-1)}\right)^l.
\end{align*}
Let $\mathcal{W}_l$ be set of all witness sequences of length $l$.
We want to count how many of these sequences can exist.
For a better approximation, recall that two consecutive members $(v_i,w_i,t_i)$ and $(v_{i+1},w_{i+1},t_{i+1})$ share an agent.
Thus, for each $i$, we can define a bit $f_i$ such that $f_i = 0$ if $v_{i+1} = w_i$ and $f_i=1$ if $v_{i=1} = v_i$.
Therefore, any member of a sequence can be expressed as $(v_i,f_i,t_i) \in P \times [1] \times [T'\tau]$ instead without losing information.
For $l \geq c\cdot 200T'$, it holds that
\begin{align*}
    |\mathcal{W}_l| &\leq \underbrace{\binom{k}{l}}_{\text{Choices of $v$}}\cdot\underbrace{\binom{T'\cdot\tau}{l}}_{\text{Choices of $t$}}\cdot\underbrace{\sum_{h=1}^l \binom{l-h}{h}}_{\text{Choices of $f$}}\\
    &\leq \binom{k}{l}\cdot\binom{T'\cdot\tau}{l}\cdot 2^l & \sum \binom{n-i}{i} = 2^n - 1\\
    &\leq \left(\frac{e \cdot k}{l}\right)^l \cdot \left(\frac{e \cdot T'\cdot\tau }{l}\right)^l \cdot 2^l & \binom{n}{i} \leq (\frac{en}{i})^i\\
    &= \left(\frac{2e^2 \cdot k \cdot T' \cdot \tau }{l}\right)^l = \left(\frac{8e^2 \cdot k \cdot T' \cdot n(n-1) }{k \cdot l}\right)^l & \tau = 4\cdot\frac{n(n-1)}{k}\\
    &= \left(\frac{8e^2 \cdot T' \cdot n(n-1) }{ l}\right)^l \leq \left(\frac{8e^2 \cdot T' \cdot n(n-1) }{ 200\cdot T'}\right)^l& l \geq 200T'\\
    &= \left(\frac{8e^2 \cdot n(n-1) }{ 200}\right)^l \leq \left(\frac{n(n-1) }{3}\right)^l & 8e^2 < 60.
\end{align*}
Finally, let $\mathcal{B}$ be the event that there is a witness sequence of length more than $c\cdot 200T'$.
We use the union bound to show
\begin{align*}
    \Prob{\mathcal{B}} &= \Prob{\bigcup_{l = c\cdot 200T'}^{T'\cdot\tau} \bigcup_{W \in \mathcal{W}_l} W } \leq \sum_{l = c\cdot 200T'}^{T'\cdot\tau} \sum_{W \in \mathcal{W}_l} \Prob{W}
    \leq  \sum_{l = c\cdot 200T'}^{T'\cdot\tau} \left(\frac{n(n-1) }{3}\right)^l\left(\frac{1}{n(n-1)}\right)^l\\
    &\leq  \sum_{l =c\cdot 200T'}^{T'\cdot\tau} \left(\frac{1}{3}\right)^l \leq T'\cdot\tau\cdot\left(\frac{1}{3}\right)^{c\cdot 200T'}
    \leq n^3 \cdot \left(\frac{1}{3}\right)^{c200 \log n} \leq \left(\frac{1}{n}\right)^{200c-3}.
\end{align*}
Therefore, a sequence of this length does not exist, w.h.p.
This proves the claim.
\end{proof}
To prove \cref{lemma:epoch_invariant}, we must show that $v$ receives a broadcast started in the epoch within the epoch.
It is easy to see that the broadcast time of broadcast message $b$ is bounded by a two-way epidemic, which is bounded by a one-way epidemic on the unranked agents.
However, we need to take into account that during the protocol's execution an unranked agent gets ranked and, therefore, stops spreading the broadcast.
Even worse, all agent holding a broadcast $b$ could get a rank and broadcast dies out completely.
As it turns, we can still show the following.
\begin{claim}
    Let $\mathcal{P} \subseteq P$ be a set of agents that are unranked at the end of the epoch.
    Then, w.h.p., all these agents received at least one broadcast within the epoch.
\end{claim}
\begin{proof}
First, note that we can assume that $|\mathcal{P}| \geq \nicefrac{k}{2}$, w.h.p.
Suppose the number of unranked agents halves at any point in the epoch.
Then, $\ell$ will start waiting (if it hasn't before).
In particular, $\ell$ will wait (and therefore not rank any agents) until it interacts with any unranked agent $\Wmax$ times.
Given that there are $k$ unranked agents initially, this takes at least $O(\frac{n(n-1)}{k} \cdot \Wmax)$ interaction on expectation and w.h.p.
For a large enough choice of $\Wmax$, this exceeds whatever time is remaining in the epoch.
Therefore, no further agents get ranked.

Going forward, we say that a time step $t$ is good if a constant fraction of $\rho\cdot|\mathcal{P}|$ agents in $\mathcal{P}$ have a coin that shows \emph{heads}.
According to \cref{lem:coin_tails}, each epoch has a constant fraction of good time steps unless there is a reset.
This follows because an episode of length $2n$ has a constant fraction of good steps with constant probability and the individual episodes are independent of one another as they assume an arbitrary distribution of the coins in the beginning of each episode.
Therefore, as the other two rules do not trigger a reset w.h.p., and during the epoch no counter gets to $0$, w.h.p., we can assume that a constant fraction of the steps is good, w.h.p.
In the remainder, we only focus on these good time steps.

Denote the event that at $\Theta(\frac{T\tau})$ time steps are good and at least $k/2$ agents stay unranked during the epoch as $\mathcal{E}$, i.e., they do not interact with $\ell$ if their coin is tails.
Let $T(\mathcal{P})$ be the number of good time steps until all agents are informed.
Let $I_t \subset \mathcal{P} \cup \{\ell\}$ be the set of agents that, in step $t$, know of a broadcast started after the beginning of the epoch.
Note that this set always contains $\ell$ as it starts all broadcasts.
Suppose that there are $i$ uninformed agents in $\mathcal{P}$ that know none of these broadcasts.
Then, the probability that the number of informed agents increases is in good step $t$ is
\begin{align*}
    \Prob{ |I_t| = |I_t|+1 \mid \{I_t = i\} \cap \{t \textit{ good }\} } &\geq \underbrace{\frac{(i-1)(m - i)}{n(n-1)}}_{\substack{ v \in I_t \cap \mathcal{P} \textit{ interacts } \\ \textit{ with }w \not\in I_t \cap \mathcal{P} }} + \underbrace{\frac{\rho(m - i)}{n(n-1)}}_{\textit{$\ell$ creates new broadcast}} \\
    &\geq \frac{\rho (i-1)(m - i) + \rho \cdot (m - i)}{n(n-1)}\\
    &\geq \rho\left(\frac{(i-1)(m - i) + (m - i)}{n(n-1)}\right)\\
    &= \rho\left(\frac{i(m - i)}{n(n-1)}\right).
\end{align*}
Furthermore, by the law of total probability it holds that
\begin{align*}
    \Prob{ |I_t| = |I_t|+1 \mid \{|I_t| = i\} \cap \{t \textit{ good }\} \cap \mathcal{E} }
    &\geq   \Prob{ |I_t| = |I_t|+1 \mid \{|I_t| = i\} \cap \{t \textit{ good }\} } - \Prob{\mathcal{E}}\\
    &\geq  \rho\left(\frac{i(m - i)}{n(n-1)}\right) - \frac{1}{n^c}
    \geq  \frac{\rho}{2}\left(\frac{i(m - i)}{n(n-1)}\right).
\end{align*}
As $\mathcal{P}$ has at least$k/2$ agents, the number of good steps required to inform all agents in $\mathcal{P}$ is stochastically dominated by
\begin{align*}
    T(\mathcal{P}) \mid \mathcal{E} \prec Y = \sum_{i \in [m-1]} X_i,\ \textup{with $X_i \sim \Geom\bc{\frac{\rho}{4} \cdot \frac{i(k-i)}{n(n-1)}}$ independent.}
\end{align*}
Therefore, we have $\Prob{T(\mathcal{P}) \geq c\cdot \rho^{-1} 24\frac{n^2}{k}\log n \mid \mathcal{E}} \leq n^{-c}$.
This means, we require $c\cdot \rho^{-1} 24\frac{n^2}{k}\log n$ good steps until all agents in $\mathcal{P}$ received the broadcast.
Thus, if we pick the length of our epoch larger than $c\cdot \rho^{-1} \cdot C \cdot \log n$ where $C$ is a large constant, it holds $c\cdot \rho^{-1} 24\frac{n^2}{k}\log n \leq \Theta({T\tau})$.
As we conditioned on $\Theta({T\tau})$ good steps in an epoch, all agents in $\mathcal{P}$ must receive a broadcast, w.h.p.
\end{proof}

Recall that we can choose $\Lmax$ as large as we want.
Thus, by choosing $\Lmax \geq 1600\cdot T$, we can combine all of our lemmas and get
\begin{align*}
    1 - n^{-c} &\leq
    \Prob{\exists b: l_b \leq c200T}
    = \Prob{\LiveCount{v} \leq \Lmax - c200T}\\
    &= \Prob{\LiveCount{v} \leq \Lmax - (\nicefrac{1}{8})\cdot \Lmax}
    = \Prob{\LiveCount{v} \leq (\nicefrac{7}{8})\cdot \Lmax} . \qedhere
\end{align*}
\end{proof}

\subsection{Auxiliary results: analysis of the phase agents' coin}

For a fixed $t$ where we assume that $\vX_t$ is a configuration in $\CM$,
    and $t' \geq t$, we
let $K_{t'}$ be the number of phase agents at time $t'$,
    and $H_{t'}$ be the number of phase agents whose coin shows $1$ at time $t'$.

\begin{lemma}\label{lem:coin_heads}
    Assume that $\vX_t$ is a configuration in $\CM$.
    Then there are constants $0 < p, c < 1$ such that
        with probability at least $p$,
            either a reset is triggered within the $2n$ rounds following $t$,
            or
            for at least a $c_1$ fraction of the next $2n$ rounds, $H_{t'} \geq K_{t'} / 4$---i.e., for $S_t = \{t' \in (t, t + 2n] \mid H_{t'} \geq K_{t'} / 4\}$ we have $|S_t| \geq c \cdot 2n$.
\end{lemma}

\begin{proof}
    For $\tau \geq 0$ consider the random variable
        \[R_{t+\tau} = \begin{cases}
            1,\quad&\textup{if for any time $t' \in (t, t+\tau]$,}\\
            &\textup{a reset was triggered at a time $t'$ or $K_{t'} = 0$,} \\
            H_t / K_t,\quad&\textup{otherwise.}
        \end{cases}\]
    We now consider the contribution to the expected change of $R$ by various events,
        noting that we want to bound this from below:\begin{itemize}
        \item In any case, if $R_{t+\tau}$ is $1$ due to the first branch of the definition, then $R_{t+\tau} = R_{t+\tau}$.
        \item If a reset is triggered or the last phase agent disappears at time $t+\tau+1$, then $R_{t+\tau+1} = 1 \geq \min\{1, R_{t+\tau+1} + 1/K_{t+\tau}\}$.
        \item If a phase agent $v$ with $\Coin v = 0$ is chosen as respondent at time $t+\tau$ (with probability $(K_{t+\tau} - H_{t+\tau})/n = (K_{t+\tau} - R_{t+\tau} K_{t+\tau}) / n$) and none of the earlier cases holds, then
        its coin is toggled and then
        $R_{t+\tau+1} = R_{t+\tau} + 1/K_{t+\tau}$.
        \item If a phase agent $v$ with $\Coin v = 1$ is chosen as respondent at time $t+\tau$ (with probability $H_{t+\tau}/n = R_{t+\tau} K_{t+\tau} / n$) and none of the earlier cases holds, then \begin{itemize}
            \item either its coin is toggled to tails and
        $R_{t+\tau+1} = R_{t+\tau} - 1/K_{t+\tau}$,
            \item or it is ranked and $R_{t+\tau+1} = (H_{t+\tau} - 1)/(K_{t+\tau} - 1) \geq R_{t+\tau} - 1/K_{t+\tau}$.
        \end{itemize}
    \end{itemize}
    Combining all of this, one can see that the additive $1/ K_{t+\tau}$ terms in the changes in $R$
        and $K_{t+\tau}$s in the denominators of the probabilities cancel
        to obtain
        \[\ExpCond{R_{t+\tau+1}}{R_{t+\tau}} \geq R_{t+\tau} + \frac{1-2R_{t+\tau}}{n} = 0.5 + (R_{t+\tau} - 0.5) - \frac{2(R_{t+\tau} - 0.5)}{n}.\]
    From this, induction over $\tau$ and using $R_{t+\tau} \geq 0$ gets us
        \[\ExpCond{0.5 - R_{t+\tau+1}}{R_t} \leq \bc{1-\frac{2}{n}} \cdot \ExpCond{0.5 - R_{t+\tau}}{R_t}
            \leq (0.5 - R_t) \cdot \bc{1-\frac{2}{n}}^\tau
            \leq \frac{1}{2} \cdot \bc{1-\frac{2}{n}}^\tau.\]
    For $\tau = 3n/2$, this yields
        \[\ExpCond{(0.5 - R_{t+3n/2})}{R_t} \leq \frac{1}{2} \cdot \exp\bc{-\frac{2\tau}{n}} = \frac{e^{-3}}{2} \leq \frac{1}{40},\]
        and hence applying Markov's inequality on $(1 - R_{t+3n/2})$ (which is $\geq 0$)
        \[\ProbCond{R_{t+3n/2} \geq 1/3}{R_t} = 1 - \Prob{1-R_{t+3n/2} \geq 2/3}{R_t} \geq 1 - \frac{0.525}{2/3} = 0.2125.\]

    Now condition on $R_{t+3n/2} \geq 1/3$.
    If this is because in the interval $(t, t+3n/2]$ a reset was triggered or there were no phase agents, we are already done.
    Otherwise, there is an $\Omega(1)$ probability that in the next $n/2$ rounds,
        at most $K_{t+3n/2}/12$ of the $H_{t+3n/2} \geq K_{t+3n/2} / 3$ phase agents $v$ having $\Coin v = 1$ at time $t+3n/2$
        are selected as respondents in the following $n/2$ rounds,
        and this is independent from the prior rounds.
    Hence, with at least that probability, from time $t+3n/2$ to time $t+2n$, at least a $\frac{1}{3} - \frac{1}{12}$ fraction of phase agents still has $\Coin{v} = 1$, as claimed.
\end{proof}

We need a similar result for the case where we want to show that no $\LivenessCount{v}$ reaches $0$, w.h.p., when the protocol is in a well-formed configuration for ranking.
There, we care about the number of rounds where there is at least a $1/4$ fraction of phase agents' coins showing tails (i.e., 0),
    so that we can ensure that some $\LivenessCount{v}$ is reset often enough.
Here, unlike above, a phase agent having its coin at $1$ being ranked doesn't work in our favor.
However, as long as we are in the good case, this will only happen with probability $\leq 1/n$ in any interaction because there is at most one unaware leader.

So for a time interval of size $2n$, there is an $\Omega(1)$ probability of the leader never being selected as initiator or responder, in which case no agent is ranked throughout the interval.

\begin{lemma}\label{lem:coin_tails}
    Assume that $\vX_t$ is a configuration in $\CM$.
    Then there are constants $0 < p, c < 1$ such that for $n$ sufficiently large,
        with probability at least $p$,
        either a reset is triggered within the $2n$ rounds following $t$, or there is more than one unaware leader in any of the $2n$ rounds following $t$,
        or for a $c_1$ fraction of the next $2n$ rounds, $H_{t'} \leq 3 K_{t'} / 4$---i.e., for $S_t = \{\tau \in (t, t + 2n] \mid H_t \leq 3 K_t / 4\}$ we have
        \(\Prob{|S_t| \geq c \cdot 2n} \geq p.\)
\end{lemma}

\begin{proof}
    For $\tau \geq 0$ consider the random variable
        \[R_{t+\tau} = \begin{cases}
            0,\quad&\textup{if for any time $t' \in (t, t+\tau]$, a reset was triggered at a time $t'$, or } \\
            &\textup{there was more than one unaware leader, or $K_{t'} = 0$,} \\
            H_t / K_t,\quad&\textup{otherwise.}
        \end{cases}\]
    We now consider the contribution to the expected change of $R$ by various events,
        noting that we want to bound this from above:\begin{itemize}
        \item In any case, if $R_{t+\tau}$ is $0$ due to the first branch of the definition, then $R_{t+\tau+1} = R_{t+\tau}$.
        \item If a reset is triggered or the last phase agent disappears at time $t+\tau+1$, then $R_{t+\tau+1} = 0 \leq \max\{0, R_{t+\tau} - 1/K_{t+\tau}\}$.
        \item If a phase agent $v$ with $\Coin v = 0$ is chosen as respondent at time $t+\tau$ (with probability $(K_{t+\tau} - H_{t+\tau})/n = (K_{t+\tau} - R_{t+\tau} K_{t+\tau}) / n$) and none of the earlier cases holds, then
        its coin is toggled and then
        $R_{t+\tau+1} = R_{t+\tau} + 1/K_{t+\tau}$.
        \item If a phase agent $v$ with $\Coin v = 1$ is chosen as respondent at time $t+\tau$ (with probability $H_{t+\tau}/n = R_{t+\tau} K_{t+\tau} / n$) and none of the earlier cases holds, then \begin{itemize}
            \item either its coin is toggled to tails and
        $R_{t+\tau+1} = R_{t+\tau} - 1/K_{t+\tau}$,
            \item or it is ranked and $R_{t+\tau+1} = (H_{t+\tau} - 1)/(K_{t+\tau} - 1) = R_{t+\tau} - \frac{1-R_{t+\tau}}{K_{t+\tau}-1} \leq R_{t+\tau}$;
                however, this can only happen with probability $\leq 1/(n-1)$ as there is at most one unaware leader in this case.
        \end{itemize}
    \end{itemize}
    Again, combining all of this, one can see that the additive $1/ K_{t+\tau}$ terms in the changes in $R$
        and $K_{t+\tau}$s in the denominators of the probabilities cancel
        to obtain
        \[\ExpCond{R_{t+\tau+1}}{R_{t+\tau}} \leq R_{t+\tau} + \frac{1-\bc{2-\frac{1}{n-1}}R_{t+\tau}}{n}
        = a + (R_{t+\tau} - a) - \frac{\bc{2-\frac{1}{n}}(R_{t+\tau} - a)}{n}.\]
    where $a = \bc{2 - \frac{1}{n}}^{-1} = \frac{n-1}{2n-3}$.
    Then, induction over $\tau$ and using $R_{t+\tau} \leq 1$ gets us
        \begin{align*}
        \ExpCond{R_{t+\tau+1}}{R_t}
           &\leq a + \bc{1-\frac{2-\frac{1}{n}}{n}} \cdot \ExpCond{R_{t+\tau} - a}{R_t}
        \leq a + (R_t-a) \cdot \bc{1-\frac{2-\frac{1}{n}}{n}}^\tau
        \leq a + (1 - a) \cdot \bc{1-\frac{2 - \frac{1}{n}}{n}}^\tau.
        \end{align*}
    Now $a = \frac{n-1}{2n-3} = \frac{1}{2} + \frac{1}{4n-6} = \frac{1}{2} + o(1)$,
        and so for $\tau = 3n/2$ we have
        \[\ExpCond{R_{t+\tau+1}}{R_t}
            \leq a + (1-a) \cdot \exp\bc{-\,\frac{(2 - o(1)\tau}{n}}
            = \frac{1}{2} + o(1) + \bc{\frac{1}{2} - o(1)} \cdot \exp(-3+o(1)),\]
        which is at most $0.525$ for sufficiently large $n$.
    An argument analogous to that at the end of the proof of \cref{lem:coin_heads} yields the claim.
\end{proof}

\fi
\end{document}